\documentclass{siamltex}

\usepackage{amsfonts} 
\usepackage{amssymb,amsmath} 
\usepackage{graphicx,color}
\usepackage[breaklinks,colorlinks=true,linkcolor=blue,citecolor=red]{hyperref}

\DeclareMathAlphabet{\itbf}{OML}{cmm}{b}{it}

\def\EE{\mathbb{E}}

\def\RR{\mathbb{R}}

\def\eps{\varepsilon}

\newcommand{\rc}{}
\newcommand{\rcc}{  }

\newcommand{\bc}{  }

\begin{document}  
\title{Correction to Black-Scholes formula due to fractional stochastic volatility}

\author{Josselin Garnier\footnotemark[1]
 and Knut S\O lna\footnotemark[2]}   

\maketitle

\renewcommand{\thefootnote}{\fnsymbol{footnote}}

\footnotetext[1]{Laboratoire de Probabilit\'es et Mod\`eles Al\'eatoires
\& Laboratoire Jacques-Louis Lions,
Universit{\'e} Paris Diderot,
75205 Paris Cedex 13,
France
{\tt garnier@math.univ-paris-diderot.fr}}

\footnotetext[2]{Department of Mathematics, 
University of California, Irvine CA 92697
{\tt ksolna@math.uci.edu}}

\renewcommand{\thefootnote}{\arabic{footnote}}

\maketitle
 
\begin{abstract}
     Empirical studies show that the volatility may exhibit
   correlations  that decay as  a fractional power of the 
   time offset. 
   The paper presents a rigorous analysis for the case
   when the stationary stochastic volatility model 
    is constructed  in terms of a fractional
    Ornstein Uhlenbeck process to have such correlations.
    It is shown how the associated implied 
    volatility  
    has  a term structure that  is a function of  maturity to a fractional power.
 \end{abstract}



\section{Introduction}

\rc{Our aim in this paper is to provide a framework for analysis
of stochastic volatility problems in the context when the volatility
process possesses correlations that decays like a power law. 
We will both consider the case of  ``long-range''  processes
where  the consecutive increments  of the process are positively 
correlated, corresponding to the so called Hurst coefficient
$H> 1/2$,  as well as the case with  ``short-range'' processes with 
consecutive increments being negatively correlated  with $H < 1/2$.}
Replacing the constant volatility of the 
Black-Scholes model with a random process gives price 
modifications in  financial contracts. It is important to understand
the qualitative behavior of such price modifications for a (class of)
stochastic volatility models since this  can be used for 
calibration purposes. Typically the price modifications are parameterized
by the implied volatility relative to the Black-Scholes model  \cite{fouque11,touzi}.
For illustration we consider here European option pricing and then 
the implied volatility depends on the moneyness, the ratio between 
the strike price and the current price, moreover, the time to maturity.
The term and moneyness  structure of the implied volatility can  be
calibrated  with respect to liquid contracts and then 
used for pricing of related but  less liquid contracts. 
Much of the  work on stochastic volatility models have focussed on situations
when the volatility process is a Markov process, commonly some sort
of a jump diffusion process. 
However, a number of empirical studies
suggest that the volatility process possesses long- and short-range 
dependence,
that is the correlation function of the volatility process has decay that is a 
fractional power of the time offset. This is the class of volatility models
we consider here.  We find that such correlations indeed reflect themself 
in an  implied volatility  fractional term structure.  An important aspect of the modeling
is also the presence of correlation between the volatility shocks and the shocks 
(driving  Brownian motion) of the underlying, this ``leverage effect''  influences
the  implied volatility  in an important way and we shall include it below. 
The leverage effect is well motivated from the modeling viewpoint 
and important to incorporate to fit  observed implied volatilities,
albeit a challenging quantity to estimate \cite{ait2}.
Evidence of leverage and persistence or long-range dependencies have been found
by considering high-frequency  data and incorporated  in   discrete time  series
models  \cite{bollerslev,engle,mykland}.

Here we model in terms 
of a continuous time stochastic volatility model that is a smooth function of a Gaussian 
process. 
We use a martingale method approach which exploits the fact
that the discounted price process is a (local) martingale. 
We model the fractional stochastic volatility (fSV) as a smooth function of a fractional 
Ornstein-Uhlenbeck  (fOU) process.  
We moreover assume that the fSV model has relatively small fluctuations,
of magnitude $\delta \ll 1$ and we derive the associated leading order
expression for implied volatility with respect to this parameter via
an asymptotic analysis.  This gives a parsimonious parameterization
of the implied volatility  which may be exploited for robust calibration.  
The fOU process
is a classic model for a stationary  process with a fractional 
correlation structure.  This process can be expressed in terms of
an integral of a fractional Brownian motion (fBm)  process. The distribution of a fBm 
 process is characterized in terms of the Hurst exponent $H \in (0,1)$.  
The fBm process is 
locally H\"older continuous of exponent $\alpha$  for all 
$\alpha  < H$ and this property is inherited by the 
fOU process.  
The fBm process, $W_t^H$,  is also self-similar in that
\begin{equation}
 \left\{
 W^H_{\alpha t} , t\in \RR  
 \right\}
 \stackrel{dist.}{=}
  \left\{
 \alpha^H W^H_{t} , t\in \RR 
 \right\}
~ \hbox{for all}~   \alpha > 0.
\end{equation}
The self-similarity property is inherited approximately by  the fOU
process on scales smaller than the mean reversion time
of the fOU process that we will denote by $1/a$ below. In this sense we
may refer to the fOU process as a multiscale process on relatively 
short scales. 

The case with $H \in (0,1/2)$ gives  
a fOU process that is a so-called 
``short-range''  dependent process that is rough on short scales and whose correlations
 for small time offsets  decay faster than the linear decay associated with a Markov process.
In fact the decay is as the offset to the fractional power
$2H$. 
In this regime consecutive  increments of the fBm process
are negatively correlated giving a rough process also referred to
as an anti-persistent process.
The enhanced negative correlation with smaller 
Hurst exponent gives a relatively rougher process.

The case with $H \in (1/2,1)$ gives  
a fOU process that is a so-called 
``long-range''  dependent process whose 
correlations for large time offsets decays as the offset to the fractional power
$2(H-1)$. It follows that the correlation function of the fOU process
is not integrable.   This regime corresponds to  a persistent
process where consecutive increments of the fBm
are positively correlated.
The relatively stronger positive correlation for the consecutive increments
of the associated fBm process with increasing $H$ values gives a relatively 
smoother process whose correlations decay relatively slowly. 
For more details regarding the fBm and fOU processes
we refer repectively to \cite{oksendal,coutinb,taqqu,mandelbrot} and \cite{cheridito03,kaarakka}. 

\rcc{In order to simplify notation and interpretation of the results we present them 
in the context of the fractional OU process. However, as we show in Appendix \ref{app:gen},  the results
 readily  generalize to the case with general Gaussian processes with short- and
long-range dependence.}

 A large number of recent papers have considered modeling of volatility in terms
 of processes with short- and long-range dependence. 
   In \cite{coutin} the authors consider a long memory extension of the Heston \cite{heston}
   option pricing model, a fractionally  integrated  square root process,
   a generalization of the early work in \cite{comte}.
   They make use of the analytical tractability of this model, 
    in fact a fractionally integrated version of a Markovian affine diffusion,  
    with affine diffusions considered in \cite{duffie}. The emphasis is on the
     long-range dependent case  ($H > 1/2$) and long time to maturity. The authors focus on the
     conditional expectation of the integrated square volatility and show the 
     fractional  decay of this, moreover, they discuss   estimation schemes for model parameters    
     based on discrete observations. In the Markovian 
     case the mean integrated square volatility
     would exponentially fast approach its mean value and flatten the implied volatility
     term structure.  They remark that long-range dependence provides an explanation for
     observations of non-flat term structure in the regime of large maturities since the long-range    
     dependence may make the implied volatility smile strongly maturity-dependent in this regime, 
     while also producing consistent smiles for short maturities. 
     The model presented in \cite{coutin} was recently revisited in \cite{jacquier} 
     where
     short and long maturity asymptotics are analyzed using large deviations principles.  
     
       The concept of RFSV, Rough Fractional Stochastic Volatility, is put forward in
       \cite{gatheral2,gatheral1}. Here a model with log-volatility modeled by
       a fBm is motivated by analysis of market data, which they state provide
       strong support
       for a value for the  Hurst exponent $H$ around 0.1.  As  explained above
       small values for $H$ correspond to very rough processes.
               It is remarked that such a  
         process can be motivated by modeling of order flow using Hawkes processes. 
         The authors discuss issues related to change from physical to pricing
         measure and use simulated prices 
          to fit well the implied volatility surface in the case of SPX
         with few parameters. They argue  that the  fractional model generates strong
         skews or ``smiles'' in the implied volatility  
          even for very short time to maturity so that
         this modeling provides an alternative to using  jumps to model such an effect.
         The form  of the implied volatility surface and the structure of the returns have
          been used to argue that the asset price
         should be a jump process
          \cite{ait,carr}.  Indeed models with jumps
           may  be used as an alternative  approach to capture smile 
          dynamics to the fractional approach considered here 
          and  recent contributions consider models driven by L\'evy
            processes both for volatility models \cite{figuer,tankov} and directly for price models 
             \cite{benth}.   
                  
        A variant of the model in \cite{gatheral1} is considered in \cite{mendes1}  where
         the log-volatility
        is modeled as a fractional noise, with fractional noise
        being the increment process of a fBm  for a certain
        increment length. The authors discuss the well-posedness of this model
        from the financial perspective and in doing so make use of a truncated version
        of the integral representation of the fBm.  
        In \cite{mendes2} this model is supported  by data analysis and motivated by 
         an agent-based interpretation.

        In \cite{viens1,viens2} the authors consider  the situation when the volatility 
        is modeled as a function of a fOU process whose
        shocks are independent from those of the underlying.  
        Their focus is on a tree-based method for computing prices,  estimation schemes
        for  model parameters, and a particle filtering technique for the unobserved 
        volatility given discrete observations.  They consider some real data examples and find
        estimated values for the Hurst exponent which is larger than $1/2$,
       in particular in a period after a market crash.   In \cite{viens3} small maturity  asymptotic results are presented for this model.   
           
   Among the many papers considering short maturity
   asymptotics,  in the early paper \cite{alos07} Al\'os et al.  use Malliavin calculus to get expressions 
   for the implied volatility in  the regime of small maturity. They find that
   the implied volatility  diverges in the short-range dependent case and flattens
   in the long-range dependent case in the limit of small maturity.
    These  results are consistent with what we present below. 
    The modeling in \cite{alos07} differs from the modeling below in that 
   the  authors consider
   volatility fluctuations at the order one level while below the fluctuations are 
   relatively small, however, we consider any time to maturity.

   Fukasawa \cite{fukasawa11}  discusses the case with small volatility fluctuations
   and short- and long-range dependence impact on the implied volatility
   as an application of the general theory he sets forth.  He uses 
   a non-stationary ``planar''  fBm as the volatility factor
    so that the leading implied volatility
   surface is identified conditioned  on the present value of the implied volatility
   factor only, while below with a stationary model the surface depends on the
   path of the volatility factor until  the present, reflecting the non-Markovian nature
   of fBm.  In \cite{fukasawa15} Fukasawa 
    discusses  the case of short-range 
    dependent processes and short time to maturity   and a framework
    for expansion of the implied volatility surface. He
    uses a representation of fBm due to Muralev \cite{muralev}.   
    He  also considers local stochastic  volatility models and find that
    these are not consistent with power laws in this regime.   
             
             As a further generalization relative to a fractional Brownian motion based model
      the case of multi fractional Brownian motion based models is considered in \cite{corley}.
       This allows  for a non-stationary local regularity or a  time dependent 
       Hurst exponent and then the implied 
       volatility depends on weighted averages of the local Hurst exponent.         
                          
       In \cite{forde} Forde and Zhang  use large deviation principles to compute the short maturity 
      asymptotic form of the implied volatility.
      They consider the correlated case with leverage and obtain results that are
      consistent with those in \cite{alos07}.
      They consider a stochastic volatility model based  on fBm
      and also more general ones where the volatility process is driven by fBms 
       and which is analyzed using rough path  theory.  They also consider 
      large time asymptotics for some fractional processes.
        
       Indeed, a number of recent papers have   considered small maturity asymptotics for 
   implied volatility in the context of mixing, short- or long-range processes.  
   Many of these use large deviation principles or  heat kernel 
   expansions \cite{beres,forde,henry},
   while another approach is to consider the regime 
    around the money \cite{alos07,fukasawa15,tankov}.
    Recent works  deal  also with the regime of large strikes and derive bounds on
    the implied volatility \cite{rodger}.
   Here we take another approach by considering a perturbation situation 
   so that the implied volatility can be expanded around an effective volatility
   \cite{fouque11}, also for large times to maturity. We  model  the volatility as a stationary process,
   a continuous time 
stationary short- or long-range dependent  stochastic volatility process,
    with a  view toward constructing a time consistent scheme. We  use an approach 
   based on the martingale method which is adapted to the 
    fact that the volatility process is not a Markov process.         
 We explicitly take into account the effects of correlation
in between volatility shocks and shocks in the underlying, the leverage effect, 
and its form in short- and long-range dependent cases.   We obtain
expressions for the implied volatility for all times to maturity and also for 
log-moneyness  of order one. 
Explicitly, we model the volatility as
\begin{equation}\label{eq:fSV}
\sigma_t =\bar{\sigma}+F(\delta Z_t^H) ,
\end{equation}
for $Z_t^H$ the fOU process that we discuss in more detail
in Section \ref{sec:fOU}. The function $F$ is assumed to be one-to-one, 
smooth, bounded from below by a constant larger than $-\bar{\sigma}$,
with bounded derivatives,
and such that $F(0)=0$ and $F'(0)=1$. \rc{
It follows that the volatility process inherits   (qualitatively) the correlation
properties of the fBm process.  Indeed,  we have
\begin{eqnarray}
\sigma_t =\bar{\sigma}+F(\delta Z_t^H) = 
\bar{\sigma} + \delta  Z^H_t  +  {\delta^2 } h^\delta(Z^H_t)
\end{eqnarray}
where $h^\delta(y)=(F(\delta y)-  \delta y)/\delta^2$  can be bounded uniformly in $\delta$ by:
\begin{eqnarray}
     |  h^\delta(y) |   \leq    \|F''\|_\infty y^2  . 
\end{eqnarray}
 Note that throughout the paper we will be working with non-dimensionalized 
quantities.   Specifically, if $t^\prime$ represents dimensionalized  time say in units
of  ``trading year''  and $T^\prime$ is a typical time horizon being for instance a typical maturity time
in years then  $t$ is the non-dimensionalized time:
\begin{eqnarray}
     t  =  \frac{t^\prime}{T^\prime} .
\end{eqnarray} }

The main result is then the associated form for  the 
implied volatility, see Equations (\ref{eq:iv1}), (\ref{eq:iv2}) and 
(\ref{eq:iv3})
  below,   we summarize the result next. 
The implied volatility is here the volatility value
that needs to be used in  the constant volatility Black-Scholes 
European option pricing  formula in order to replicate 
the asymptotic fSV option price,  it is, up to terms of order $\delta^2$:
\begin{equation}
\label{eq:main}
I_t = \EE\Big[ \frac{1}{T-t} 
\int_t^T \sigma_s^2 ds |{\cal F}_t \Big]^{\frac{1}{2}} 
+  {\cal A}(T-t) \Big[ 1  + \frac{\log(K/X_t)}{(T-t) / \bar\tau} \Big]
   ,
\end{equation}
for 
\begin{equation}
{\cal A}(\tau) =
 \frac{\delta \rho \bar{\sigma}  \tau^{H+\frac{1}{2}}}{2 \Gamma(H+\frac{5}{2})}
\Big\{ 1-\int_0^{a \tau} e^{-v} \big( 1-\frac{v}{a\tau} \big)^{H+\frac{3}{2}}  
dv \Big\}  , 
\end{equation}
 where  $1/a$ is the mean reversion time of the fOU process 
and $\bar\tau=2/\bar\sigma^2$ a characteristic diffusion time for the 
underlying. 
Furthermore,  $X_t$ is  the underlying price process with evolution as
in (\ref{eq:X}) and   ${\cal F}_t$ its associated filtration. 
Moreover,  $\rho$ is the correlation in between the Brownian motions
driving respectively the volatility process and the underlying price 
process,  $K$ is the strike price so that $K/X$ is the moneyness,
and finally $\tau=T-t$ is time to maturity.   
The first term  in the implied volatility is  
the expected effective volatility over the remaining  time
period  of the option conditioned on the   knowledge
at time $t$, note that this term is random. 
The second term is a leverage term which is present 
in the case that  the underlying and the volatility have  correlated
evolutions  so that $\rho$ is non-zero. Note that $\rho$ is commonly
assumed to be negative.  The  log-moneyness term becomes
relatively more important as the time to maturity becomes small
relative to the characteristic diffusion time. 

In   the short  and long  time to maturity 
regimes we then have for the leverage term:
\begin{equation}\label{array}
 {\cal A}(\tau) \Big[ 1  + \frac{\log(K/X_t)}{\tau / \bar\tau} \Big]
 = \left\{
\begin{array}{ll}
  a_{\rm s}
 \Big[  (\tau/\bar\tau)^{\frac{1}{2}+H}  + 
  (\tau/\bar\tau)^{-\frac{1}{2}+H} 
 {\log(K/X_t)}  \Big]  &   ~\hbox{for}~ a\tau \ll 1  ,  \\
  a_{\rm l}
 \Big[  (\tau/\bar\tau)^{-\frac{1}{2}+H}  + 
  (\tau/\bar\tau)^{-\frac{3}{2}+H} 
 {\log(K/X_t)}  \Big]  &   ~\hbox{for}~ a\tau \gg  1  \, , 
\end{array}
\right.
\end{equation}
with
\begin{equation} 
 a_{\rm s}  = \frac{ \delta  \rho \bar\tau^{H} }{ \sqrt{2} \Gamma(H+\frac{5}{2})} , \quad \quad 
  a_{\rm l} =      \frac{ \delta  \rho \bar\tau^{H-1} }{ \sqrt{2} a \Gamma(H+\frac{3}{2})}  .
\end{equation}
 
We moreover  have for the predicted effective volatility term:
\begin{equation}
 \sigma_{t,T}   \equiv \EE \Big[  \frac{1}{T-t}  \int_t^T \sigma_s^2 ds |{\cal F}_t \Big]^{\frac{1}{2}}   
    = \left\{
\begin{array}{ll}
\sigma_t    &   ~\hbox{for}~ a\tau \ll 1  , \\
\bar\sigma   &   ~\hbox{for}~ a\tau \gg 1 .  
\end{array}
\right. 
\end{equation}
It is important to note that we only assume $\tau =T-t>0$ so that 
in fact the implied volatility for small times to maturity
may be very large for short-range dependent processes. 
This reflects the fact that for short-range dependent processes the volatility path is 
rough and may have a significant impact beyond the current predicted effective 
volatility level.
However, when used in the standard Black-Scholes pricing formula the implied
volatility indeed gives a pricing correction that is $O(\delta)$ for any 
$\tau>0$.   We also note that in the long maturity regime the implied 
volatility level may diverge for long-range dependent processes reflecting 
the fact that long-range dependence gives strong temporal coherence and therefore 
relatively large corrections to the predicted current effective volatility.    

Note next that the calibration of  the 
leverage component of the implied volatility  in the general case
in (\ref{eq:main})
 involves estimation of the group  market parameters:
\begin{eqnarray}
\label{eq:cal1} 
    \bar{\sigma} , \quad H, \quad (\delta \rho), \quad  a  , 
\end{eqnarray} 
 from observed implied volatility data.
 In order to fully identify the model at the current time $t$ 
we need moreover to estimate  the current predicted effective 
volatility over a time to maturity horizon, that is,
$\sigma_{t,t+\tau} $ for
    $0 \leq \tau \leq T_{\rm max}-t$.
  
It is important to note that in our framework the  market parameters
are from the theoretical point of view independent of the current time
$t$.   Thus, in order to calibrate the model with data over a current time
epoch  $  t_1 \leq t \leq t_2$ one may  use all the 
implied volatility recording  
in a joint fitting procedure.   

%
%
%
  
 We remark  that  our results would be modified 
under the presence of general interest rates and market price of risk factors
that we do not consider here.  We also remark that 
identifying a   ``smile'' shape,  that is a more general function in   log-moneyness,
would require  
a higher-order approximation of implied volatility  
\cite{smile}.  Finally,  observe  that the case $H=1/2$ corresponds neither
to a short-range dependent process nor a long-range dependent process,
but the standard case of an Ornstein-Uhlenbeck process and a stochastic
volatility that is a Markovian process with correlations decaying exponentially fast
\cite{fouque11}.

The framework we have presented is general and can be used for processes
for which we can identify the key quantities of interest below. 
We discuss one important special case corresponding to a slow fOU process.
In this case we model the  volatility in terms of the ``slow''  fOU process $Z^{\delta,H}$:
\begin{equation}
Z_{t}^{\delta,H}  = \delta^H \int_{-\infty}^t e^{-\delta a (t-s)} dW^H_s ,
\end{equation}
whose natural time scale is $1/\delta$ 
and whose  variance is order one and given 
by $\sigma_{\rm ou}$ defined by (\ref{def:sigmaou}) below,
 independently of $\delta$.
Then the volatility is
\begin{equation}\label{eq:fSVslow0}
\sigma_t = F( Z_{t}^{\delta,H} ) ,
\end{equation}
where $F$ is a smooth, positive-valued function, bounded away from zero, with bounded derivatives. 
We introduce the two parameters
\begin{equation}\label{eq:pars}
\sigma_0=F( Z_{0}^{\delta,H} ), \quad\quad p_0=F'(Z_{0}^{\delta,H}) ,
\end{equation}
that is, the local level and rate of change of the volatility.
In this case the implied volatility is given by:
\begin{equation}\label{eq:Is}
I_t = \EE\Big[ \frac{1}{T-t} \int_t^T \sigma_s^2 ds |{\cal F}_t \Big]^{\frac{1}{2}} 
+  \frac{\delta^H p_0 \rho \tau_0^H  }{ \sqrt{2} \Gamma(H+\frac{5}{2})}
\Big[ 
 (\tau/\tau_0)^{\frac{1}{2}+H}  + 
  (\tau/\tau_0)^{-\frac{1}{2}+H} 
 {\log(K/X_t)}   
 \Big]
   ,
\end{equation}
 for $\tau_0=2/\sigma_0^2$.  Thus, the slow fractional volatility factor
 yields an implied volatility that   corresponds to the one of the fractional  model in (\ref{eq:fSV})   in the regime of  small maturity,  as given in (\ref{array}).
In the special case that $H=1/2$ the volatility process becomes  a standard Ornstein-Uhlenbeck
process and is in  the class of slow processes considered in
\cite{fouque11} and indeed the implied volatility  in (\ref{eq:Is}) can then be show to be exactly
of the form discussed for the slow correction in \cite{fouque11} (Chapter 5).

The outline of the paper is as follows.
First in Section \ref{sec:fSV} 
we introduce   the details of the  ingredients of the 
fSV model.
In Section \ref{sec:option} we derive the main result of
the paper, the leading order expression for the price
in the situation with a fSV.
The derivation  is based on a contract  with a smooth 
payoff function while the European payoff function 
has a kink singularity and we generalize the result
to this situation in Section \ref{sec:euro}.
Then in Section \ref{sec:implied} we derive the 
expression for the implied volatility and how the
fractional character of the volatility affects this.
We connect  to the  slow time volatility model  in Section \ref{sec:conn} and 
present some concluding remarks in Section
\ref{sec:concl}.
In Appendix \ref{sec:tech} we characterize   some quantities
of interest and associated technical lemmas that are being used
in the price derivation in Section \ref{sec:option}.

\section{The fractional stochastic volatility model}
\label{sec:fSV}
  
We describe in more detail the fBm and fOU processes that
are used in the fSV construction (\ref{eq:fSV}).
  
\subsection{Fractional Brownian motion and its moving-average stochastic integral representation}

A fractional Brownian motion (fBm) is a zero-mean Gaussian process $(W^H_t)_{t\in \RR}$  
with the covariance
\begin{equation}
\label{eq:covfBm}
\EE[ W^H_t W^H_s ] = \frac{\sigma^2_H}{2} \big( |t|^{2H} + |s|^{2H} - |t-s|^{2H} \big) ,
\end{equation}
where $\sigma_H$ is a positive constant.

We use the following moving-average stochastic integral representation of the
fBm \cite{mandelbrot}:
\begin{equation}\label{eq:Wd}
W^H_t = \frac{1}{\Gamma(H+\frac{1}{2})} 
\int_{\RR} (t-s)_+^{H - \frac{1}{2}} -(-s)_+^{H - \frac{1}{2}} dW_s ,
\end{equation}
where $(W_t)_{t \in \RR}$ is a standard Brownian motion over $\RR$.
In this model $(W^H_t)_{t\in \RR}$ is a zero-mean Gaussian process with
the covariance (\ref{eq:covfBm}) where
\begin{eqnarray}
\nonumber
\sigma^2_H &=& \frac{1}{\Gamma(H+\frac{1}{2})^2} 
\Big[ \int_0^\infty \big( (1+s)^{H - \frac{1}{2}} -s^{H - \frac{1}{2}} \big)^ 2 ds +\frac{1}{2H}\Big] \\
&=&  \frac{1}{\Gamma(2H+1) \sin(\pi H)}.
\end{eqnarray}

\subsection{The fractional Ornstein-Uhlenbeck process}
\label{sec:fOU}

We then introduce the fractional Ornstein-Uhlenbeck process (fOU) as 
\begin{equation}\label{eq:Zd}
Z^H_t = \int_{-\infty}^t e^{-a(t-s)} dW^H_s  = W^H_t - a \int_{-\infty}^t e^{-a(t-s)} W^H_s ds.
\end{equation}
It is a zero-mean, stationary Gaussian process,
with variance
\begin{equation}
\label{def:sigmaou}
\sigma^2_{{\rm ou}} =
\EE [ (Z^H_t)^2 ] = \frac{1}{2} a^{-2H} \Gamma(2H+1)  \sigma^2_H ,
\end{equation}
and covariance:
\begin{eqnarray}
\nonumber
\EE [ Z^H_t Z^H_{t+s}  ] &=&
 \sigma^2_{{\rm ou}} \frac{1}{\Gamma (2H+1)}
 \Big[ \frac{1}{2} \int_\RR e^{- |v|} | as+v|^{2H} dv - |as|^{2H}\Big]
\\
&=&
 \sigma^2_{{\rm ou}} \frac{2 \sin (\pi H)}{\pi}
\int_0^\infty \cos( a s x) \frac{x^{1-2H}}{1+x^2} dx .
\end{eqnarray}
Note that  it is not a martingale, neither a Markov process.

\rc{
Substituting  (\ref{eq:Wd})  into the 
second representation  in Eq.~(\ref{eq:Zd}) gives  in view of  stochastic Fubini
 the moving-average integral representation of the fOU:}
\begin{equation}\label{eq:eqOU}
Z^H_t = \int_{-\infty}^t {\cal K}(t-s) dW_s   ,
\end{equation}
where 
\begin{eqnarray}
\label{def:calK}
{\cal K}(t) &=& \frac{1}{\Gamma(H+\frac{1}{2})} 
 \Big[ t^{H - \frac{1}{2}} - a\int_0^t (t-s)^{H - \frac{1}{2}} e^{-as} ds \Big] 
  .
\end{eqnarray}

The properties of the kernel ${\cal K}$ are the following ones:\\
- ${\cal K}$ is nonnegative-valued,
${\cal K} \in L^2(0,\infty)$ for any $H \in (0,1)$ with $\int_0^\infty {\cal K}^2(u) du = \sigma_{\rm ou}^2$, 
and ${\cal K} \in L^1(0,\infty)$ for any $H \in (0,1/2)$.\\
- For small times $at \ll 1$:
\begin{equation}
\label{eq:asyK1}
{\cal K} (t) = \frac{1}{\Gamma(H+\frac{1}{2}) a^{H - \frac{1}{2}} } 
\Big( (at)^{H - \frac{1}{2}}  +o\big(  (at)^{H - \frac{1}{2}} \big) \Big) .
\end{equation}
- For large times $at \gg 1$:
\begin{equation}
\label{eq:asyK2}
{\cal K} (t) = \frac{1}{\Gamma(H-\frac{1}{2}) a^{H - \frac{1}{2}}}  
\Big( (at)^{H - \frac{3}{2}}  +o\big(  (at)^{H - \frac{3}{2}} \big) \Big) .
\end{equation}
For $H \in (0,1/2) $
 the fOU process  possesses short-range correlation properties:
\begin{equation}
\label{eq:corZG2}
\EE [ Z^H_t Z^H_{t+s}  ]  =  \sigma^2_{{\rm ou}}  \Big( 
1 - \frac{1}{\Gamma(2H+1)} (a s)^{2H} + o\big( (as)^{2H}\big) \Big), \quad \quad as \ll 1.
\end{equation}
For $H \in  (1/2,1)$
 it possesses long-range correlation properties:
\begin{equation}
\label{eq:corZG3}
\EE [ Z^H_t Z^H_{t+s}  ]  =  \sigma^2_{{\rm ou}}  
\Big(
 \frac{1}{  \Gamma(2H-1)}  (as)^{2H-2}
+ o\big( (as)^{2H-2}\big) \Big)
, \quad \quad as \gg 1.
\end{equation}
The expansion (\ref{eq:corZG3}) is valid for any $H \in  (0,1/2) \cup (1/2,1)$
and for $H \in  (1/2,1)$ it shows that the correlation function is not integrable at infinity.
This is in contrast to the case of short-range dependent processes and also to
Markov processes for which the correlation function is integrable.

\section{The option price}
\label{sec:option}
The price of the risky asset follows the stochastic differential equation:
\begin{equation}
\label{eq:X}
dX_t = \sigma_t X_t dW^*_t ,
\end{equation}
where the stochastic volatility is
\begin{equation}\label{eq:svm}
\sigma_t =\bar{\sigma}+F(\delta Z_t^H) ,
\end{equation}
$Z_t^H$ has been introduced in the previous section
and is adapted to the Brownian motion $W_t$,
and $W^*_t$ is a Brownian motion that is correlated to the stochastic volatility through
\begin{equation}\label{eq:lev}
W^*_t = \rho W_t + \sqrt{1-\rho^2}
B_t ,
\end{equation}
 where the Brownian motion $B_t$ is independent of $W_t$.
\rc{We remark that the main aspect of the  model whose consequences we want to analyze
here are  the short- respectively long-range properties of the correlation function in
Eqs. (\ref{eq:corZG2}) and (\ref{eq:corZG3}) under the presence
of leverage as in Eq. (\ref{eq:lev}). We will find that this has a dramatic effect
on the asymptotic prices and the associated implied volatility. }

The function $F$ is assumed to be one-to-one, 
smooth, bounded from below by a constant larger than $-\bar{\sigma}$,
\rc{bounded above}, 
with bounded derivatives, and such that $F(0)=0$ and $F'(0)=1$.
\rcc{Note that with this normalization for  the function $F$ it will not appear 
explicitly in the price approximation  in Proposition \ref{prop:0} below as further properties
are not important in that context.} 
Moreover, then the filtration ${\cal F}_t $ generated
by $(B_t,W_t)$ is also the one generated by $X_t$.
Indeed, it is equivalent to the one generated by $(W^*_t,W_t)$, or $(W^*_t,Z^H_t)$.
Since $F$ is one-to-one, it is equivalent to the one generated by $(W^*_t,\sigma_t)$.
Since $\bar{\sigma}+F$ is positive-valued,
it is equivalent to the one generated by $(W^*_t,\sigma_t^2)$, or $X_t$.

We aim at computing the option price defined as the martingale
\begin{equation}
\label{def:optionprice}
M_t =\EE\big[ h(X_T) |{\cal F}_t \big]  ,
\end{equation}  
where $h$ is a smooth  function with bounded derivatives apart from
a finite set of points where it may have a jump discontinuity in its derivative. 
Note that the proof in this section will be given for the case  with $h$ smooth and
bounded. 
However, as we only need to control the function $Q^{(0)}_t(x)$ defined below rather than $h$
the argument can be extended from the smooth  case to the situation with jump discontinuity in 
the derivative which is relevant in the case with a call payoff.
 We carry out this generalization explicitly 
in  Section \ref{sec:euro}.  

The idea of the proof that we present below is to construct an approximation
for $M_t$ which has the correct terminal condition and which up to small
(order $\delta^2$) terms is a martingale. It then follows that we have
a price approximation to $O(\delta^2)$. 

We introduce the operator 
\begin{equation}
{\cal L}_{\rm BS} (\sigma) = \partial_t +\frac{1}{2} \sigma^2 x^2 \partial_x^2 .
\end{equation}

The following proposition gives the first-order correction to the 
expression of the martingale $M_t$ when $\delta$ is small.
\begin{proposition}
\label{prop:0}%
When $\delta$ is small,
we have
\begin{equation}
M_t = Q_t(X_t) +O(\delta^2) ,
\end{equation}
where
\begin{equation}
Q_t(x) = Q_t^{(0)}(x) +\delta \bar{\sigma} \phi_t \big( x^2\partial_x^2 Q_t^{(0)}(x)\big)
+\delta \rho Q_t^{(1)}(x) ,
\end{equation}
$Q_t^{(0)}(x) $ is deterministic and given by the Black-Scholes formula with constant volatility~$\bar{\sigma}$,
\begin{equation}
\label{eq:bs0}
{\cal L}_{\rm BS} (\bar\sigma) Q_t^{(0)}(x) =0,   \quad \quad Q_T^{(0)}(x) = h(x),
\end{equation}
$\phi_t$ is the random component
\begin{equation}
\label{def:phit}
\phi_t = \EE \Big[\int_t^T Z_s^H ds |{\cal F}_t \Big],
\end{equation}
and $Q_t^{(1)}(x) $ is the deterministic correction
\begin{equation}
\label{def:Q1t}
Q_t^{(1)}(x) = \bar{\sigma}^2 x \partial_x \big(x^2 \partial_x^2 Q_t^{(0)}(x)  \big) D_{t,T} ,
\end{equation}
with $D_{t,T}$ defined by 
\begin{equation}
 D_{t,T} =   {\cal D}(T-t), \quad \quad {\cal D}(\tau) = \frac{\tau^{H+\frac{3}{2}}}{\Gamma(H+\frac{5}{2})}
\Big\{ 1-\int_0^{a \tau} e^{-v} \big( 1-\frac{v}{a\tau} \big)^{H+\frac{3}{2}}  dv \Big\}  .
\label{eq:represcalD}
\end{equation}
\end{proposition}

\bc{ 
The correction $Q_t^{(1)}$ solves the problem in (\ref{eq:q1}) below. The function  ${\cal D}(\tau)$ 
derives from solving this problem and is:
\[
{\cal D}(\tau) =  \int_0^\tau (\tau-u){\cal K}(u) du,
\]
it  is discussed in more detail in  Lemma \ref{lem:2}  in Appendix  \ref{sec:tech}.
}

Note that the stochastic volatility process we 
have introduced in Eq.  (\ref{eq:svm}) is a stationary power-law process.
As a consequence of our modeling we have  in particular that $\phi_t$ 
is a Gaussian ${\cal F}_t$-measurable process, it   reflects the 
influence of the past on the future stochastic  volatility path conditioned
on the present.    We next present the proof of Proposition  \ref{prop:0} and remark
that in the analytic framework that we set forth, exploiting the ``$\eps$-martingale
decomposition''  \cite{fouque11}, the cases with $H< 1/2$ and $H> 1/2$ can be treated in a uniform way. 

\bc{ The proof we present below holds for general Gaussian processes $Z_t$ with short- and long-range correlations,
while the expression to the right  in Eq. (\ref{eq:represcalD}) is specific  to the fOU process.
We discuss the general Gaussian case and the general expression (\ref{def:calDB}) of 
${\cal D}(\tau)$ in Appendix~\ref{app:gen}.  
}
 
\begin{proof}
For any smooth function $q_t(x)$, we have by It\^o's formula
\begin{eqnarray*}
dq_t(X_t) &=& \partial_t q_t(X_t) dt +  \big( x \partial_x q_t\big) (X_t) \sigma_t dW_t^*
+\frac{1}{2}  \big( x^2 \partial_x^2 q_t\big) (X_t) \sigma_t^2 dt\\
&=&
{\cal L}_{\rm BS}(\sigma_t) q_t(X_t) dt +  \big( x \partial_x q_t\big) (X_t) \sigma_t dW_t^*  ,
\end{eqnarray*}
the last term being a martingale.
Therefore, by (\ref{eq:bs0}), we have
\begin{eqnarray}\label{eq:Q0}
dQ_t^{(0)}(X_t)  &=&
\big( \delta\bar{\sigma}  Z^H_t + \frac{\delta^2 }{2} g^\delta(Z^H_t)\big) \big( x^2 \partial_x^2 \big) Q_t^{(0)}(X_t) dt + dN^{(0)}_t ,
\end{eqnarray}
with $N_t^{(0)}$ a martingale,
$$
dN_t^{(0)}= \big( x \partial_x Q_t^{(0)}\big) (X_t) \sigma_t dW_t^* ,
$$
and $g^\delta(y)$ is the function
\begin{eqnarray*}
g^\delta(y) = 2 \bar{\sigma} \frac{F(\delta y)-\delta y}{\delta^2} +\frac{F(\delta y)^2}{\delta^2} ,
\end{eqnarray*}
that can bounded uniformly in $\delta$ by
\begin{eqnarray*}
|g^\delta(y)| \leq \big( \bar{\sigma} \|F''\|_\infty + \|F'\|_\infty^2 \big) y^2.
\end{eqnarray*}  \rc{
Note also that in Eq. (\ref{eq:Q0})  (and below) we use the notation
\begin{eqnarray*}
\big( x^2 \partial_x^2 \big) Q_t^{(0)}(X_t) = \left(  \big( x^2 \partial_x^2 \big) Q_t^{(0)}(x) \right)\big|_{ x=X_t}  . 
\end{eqnarray*} }

Let $\phi_t$ be defined by (\ref{def:phit}). We have
\begin{equation}\label{eq:phi}
\phi_t = \psi_t - \int_0^t Z^H_s ds ,
\end{equation} 
where the martingale $\psi_t$ is defined by 
\begin{equation}
\label{def:Kt}
\psi_t = \EE \Big[ \int_0^T Z_s^H ds | {\cal F}_t\Big] ,
\end{equation}
and it is studied in Appendix \ref{sec:tech}.
We can write
$$
Z^H_t  \big( x^2 \partial_x^2 \big) Q_t^{(0)}(X_t) dt =
 \big( x^2 \partial_x^2 \big) Q_t^{(0)}(X_t) d\psi_t -
 \big( x^2 \partial_x^2 \big) Q_t^{(0)}(X_t) d\phi_t .
$$
By It\^o's formula:
\begin{eqnarray*}
d \big( \phi_t \big( x^2 \partial_x^2 \big) Q_t^{(0)}(X_t)\big) &=&
\big( x^2 \partial_x^2 \big) Q_t^{(0)}(X_t) d\phi_t\\
&&+
\big( x\partial_x\big(  x^2 \partial_x^2 \big)\big) Q_t^{(0)}(X_t) \sigma_t \phi_t dW_t^*
\\
&&+
\frac{1}{2}\big( x^2\partial_x^2\big(  x^2 \partial_x^2 \big)\big) Q_t^{(0)}(X_t) \sigma_t^2 \phi_t dt
\\
&&+
 \big(  x^2 \partial_x^2  \partial_t \big) Q_t^{(0)}(X_t)  \phi_t dt
\\
&&+ \big( x\partial_x\big(  x^2 \partial_x^2 \big)\big) Q_t^{(0)}(X_t) \sigma_t d\left< \phi ,W^*\right>_t \\
&=&\big( x^2 \partial_x^2 \big) Q_t^{(0)}(X_t) d\phi_t\\
&&+
\big( x\partial_x\big(  x^2 \partial_x^2 \big)\big) Q_t^{(0)}(X_t) \sigma_t \phi_t dW_t^*
\\
&&+
\big( \delta \bar{\sigma} Z^H_t+\frac{1}{2}\delta^2 g^\delta(Z^H_t)\big) 
\big( x^2\partial_x^2\big(  x^2 \partial_x^2 \big)\big) Q_t^{(0)}(X_t) \phi_t dt
\\
&&+ \big( x\partial_x\big(  x^2 \partial_x^2 \big) \big) Q_t^{(0)}(X_t) \sigma_t d\left< \phi ,W^*\right>_t   ,
\end{eqnarray*}
where we have used again ${\cal L}_{\rm BS} (\bar\sigma) Q_t^{(0)}(x) =0$.
We have $\left< \phi ,W^*\right>_t = \rho \left< \psi ,W\right>_t$ and therefore
\begin{eqnarray*}
d \big( \phi_t \big( x^2 \partial_x^2 \big) Q_t^{(0)}(X_t)\big)
&=& - Z^H_t  \big( x^2 \partial_x^2 \big) Q_t^{(0)}(X_t) dt \\
&&+
\big( \delta \bar{\sigma} Z^H_t+\frac{1}{2}\delta^2 g^\delta(Z^H_t)\big) 
\big( x^2\partial_x^2\big(  x^2 \partial_x^2 \big)\big) Q_t^{(0)}(X_t) \phi_t dt
\\
&&+ \rho \big( x\partial_x\big(  x^2 \partial_x^2 \big) \big) Q_t^{(0)}(X_t) \sigma_t d\left<\psi ,W\right>_t\\
&&+ dN^{(1)}_t ,
\end{eqnarray*}
where $N^{(1)}_t$ is a martingale,
$$
dN^{(1)}_t =\big( x\partial_x\big(  x^2 \partial_x^2 \big)\big) Q_t^{(0)}(X_t) \sigma_t \phi_t dW_t^*
+
 \big( x^2 \partial_x^2 \big) Q_t^{(0)}(X_t) d\psi_t .
$$
Therefore:
\begin{eqnarray}
\nonumber
&&d \big( Q_t^{(0)}(X_t)+\delta \bar{\sigma}\phi_t \big( x^2 \partial_x^2 \big) Q_t^{(0)}(X_t)\big) \\
\nonumber
&&=  \big( \delta^2 \bar{\sigma}^2 Z^H_t+\frac{1}{2}\delta^3\bar{\sigma} g^\delta(Z^H_t)\big) 
\big( x^2\partial_x^2\big(  x^2 \partial_x^2 \big)\big) Q_t^{(0)}(X_t) \phi_t dt
\\
\nonumber
&& \quad +\frac{\delta^2 }{2} g^\delta(Z^H_t)  \big( x^2 \partial_x^2 \big) Q_t^{(0)}(X_t) dt 
+ \delta \bar{\sigma}\rho \big( x\partial_x\big(  x^2 \partial_x^2 \big) \big) Q_t^{(0)}(X_t) \sigma_t d\left<\psi ,W\right>_t\\
&&\quad + 
dN^{(0)}_t +\bar{\sigma}\delta dN^{(1)}_t  .
\end{eqnarray}

The deterministic function $Q^{(1)}_t$ defined by (\ref{def:Q1t}) satisfies
\begin{eqnarray}\label{eq:q1}
{\cal L}_{\rm BS}(\bar{\sigma}) Q^{(1)}_t(x) = -\bar{\sigma}^2
\big( x\partial_x \big( x^2 \partial_x^2 Q^{(0)}_t(x)\big)\big) \theta_{t,T}, \quad \quad Q^{(1)}_T(x) = 0,
\end{eqnarray}
where $\theta_{t,T}$ is such that 
$$
d\left< \psi,W\right>_t =\theta_{t,T} dt,
$$
and it is given by (see Lemma \ref{lem:1}):
\begin{equation}
\label{def:thetatT}
\theta_{t,T} = 
 \int_t^T {\cal K}(v-t) dv = \int_0^{T-t}{\cal K}(v) dv .
\end{equation}
Applying It\^o's formula
\begin{eqnarray*}
dQ_t^{(1)}(X_t)  
&=&
{\cal L}_{\rm BS}(\sigma_t) Q_t^{(1)}(X_t) dt +  \big( x \partial_x Q_t^{(1)}\big) (X_t) \sigma_t dW_t^* \\
&=& {\cal L}_{\rm BS}(\bar\sigma) Q_t^{(1)}(X_t) dt + 
\big( \delta\bar{\sigma}  Z^H_t + \frac{\delta^2 }{2} g^\delta(Z^H_t)\big) \big( x^2 \partial_x^2 \big) Q_t^{(1)}(X_t) dt 
\\&&+
 \big( x \partial_x Q_t^{(1)}\big) (X_t) \sigma_t dW_t^*  \\
 &=& -\bar{\sigma}^2
\big( x\partial_x \big( x^2 \partial_x^2 \big)\big)Q^{(0)}_t (X_t) d\left< \psi,W\right>_t  \\&& +
\big( \delta\bar{\sigma}  Z^H_t
+ \frac{\delta^2 }{2} g^\delta(Z^H_t)\big) \big( x^2 \partial_x^2 \big) Q_t^{(1)}(X_t) dt + dN^{(2)}_t ,
\end{eqnarray*}
where $N^{(2)}_t$ is a martingale,
$$
dN^{(2)}_t =  \big( x \partial_x Q_t^{(1)}\big) (X_t) \sigma_t dW_t^*   .
$$
Therefore
\begin{eqnarray}
d \big( Q_t^{(0)}(X_t)+\delta \bar{\sigma}\phi_t \big( x^2 \partial_x^2 \big) Q_t^{(0)}(X_t)
+ \delta \rho Q_t^{(1)}(X_t)  
\big)
= dN_t -d R_{t,T} ,
\end{eqnarray}
where $N_t$ is a martingale,
\begin{eqnarray}
\label{def:N4}
N_t = \int_0^t dN^{(0)}_s +\bar{\sigma}\delta dN^{(1)}_s   +  \rho \delta dN^{(2)}_s ,
\end{eqnarray}
and $R_{t,T}$ is of order $\delta^2$:
\begin{eqnarray}
\nonumber
R_{t,T} 
&=& \delta^2  \int_t^T \big(  \bar{\sigma}^2 Z^H_s+\frac{1}{2}\delta \bar{\sigma} g^\delta(Z^H_s)\big) 
\big( x^2\partial_x^2\big(  x^2 \partial_x^2 \big)\big) Q_s^{(0)}(X_s) \phi_s ds
\\
\nonumber
&&  + \frac{\delta^2 }{2}  \int_t^T g^\delta(Z^H_s)  \big( x^2 \partial_x^2 \big) Q_s^{(0)}(X_s) ds
 + \delta^2  \int_t^T  \bar{\sigma}\rho \big( x\partial_x\big(  x^2 \partial_x^2 \big) \big) Q_s^{(0)}(X_s) Z^H_s
 \theta_{s,T} ds
 \\
&&
+ \delta^2  \int_t^T
\big(  \rho \bar{\sigma}  Z^H_s+ \frac{\delta }{2} \rho g^\delta(Z^H_s)\big)
 \big( x^2 \partial_x^2 \big) Q_s^{(1)}(X_s) ds .
\label{def:R4}
\end{eqnarray}

Then with $Q_t(x)$ defined as in Proposition \ref{prop:0} we have
  $Q_T(x) = h(x)$ because $Q_T^{(0)}(x)=h(x)$, $\phi_T=0$, and $Q^{(1)}_T(x)=0$.
Therefore
\begin{eqnarray}
\nonumber
M_t &=& \EE \big[ h(X_T) |{\cal F}_t \big] = 
\EE \big[ Q_T(X_T) |{\cal F}_t \big] 
=
Q_t(X_t) + \EE \big[ N_T-N_t | {\cal  F}_t \big]+
\EE \big[ R_{t,T} |{\cal F}_t \big]\\
&=&
Q_t(X_t)  +
\EE \big[ R_{t,T} |{\cal F}_t \big]
 ,
\end{eqnarray}
which completes the proof since $\EE \big[ R_{t,T} |{\cal F}_t \big]$ is of order $\delta^2$.
\end{proof}

\section{Accuracy with European option}
\label{sec:euro}
In the derivation above we assumed a smooth payoff function.
Since  important classes of payoff functions have non-smooth payoff
we generalize here the proof to such a class by considering a European 
option. 
For a European option $h(x)=(x-K)_+$ we have from Eq.   1.41  in \cite{fouque11}   
\begin{eqnarray}
\nonumber
Q^{(0)}_t(x) &=& 
x \Phi
\Big(\frac{1}{\bar{\sigma}\sqrt{T-t}}  \log \big(\frac{x}{K}\big) + \frac{\bar{\sigma} \sqrt{T-t} }{2}  \Big)
\\
&&- 
K \Phi
\Big(\frac{1}{\bar{\sigma}\sqrt{T-t}}   \log \big(\frac{x}{K}\big)- \frac{\bar{\sigma} \sqrt{T-t} }{2} \Big),
\label{def:Q0e}
\end{eqnarray}
where $\Phi$ is the cumulative distribution function of the standard normal distribution.
We can see that $h$ is not smooth so that the hypotheses of Proposition \ref{prop:0} are not satisfied.
However the conclusions of Proposition \ref{prop:0} still hold true 
as we now show.

\begin{proof}
One has to show that $R_{t,T}$ defined by (\ref{def:R4})
satisfies $\EE \big[ R_{t,T}|{\cal F}_t \big]$ is of order $\delta^2$ in $L^p$ 
for any $p$ and that the local martingale $N_t$ defined by (\ref{def:N4}) is a martingale (up to time $T$). 
The problem comes from the fact that the derivatives of $Q^{(0)}_t(x)$ blow up 
when $t \to T$. However this blow up is not strong as we show below.
We first state a few properties of the deterministic and random terms that appear in the expression
of $R_{t,T}$:\\
- The deterministic function $Q^{(0)}_t(x)$ given by (\ref{def:Q0e}) satisfies
$$
 \partial_x^k Q^{(0)}_t(x) \leq C \Big( 1 + \frac{1}{(T-t)^{\frac{k-1}{2}}}\Big) ,
$$
for any $1\leq k \leq 4$, $t \in [0,T]$, $x \in (0,\infty)$, and for some constant $C$ \rc{(see Appendix B in \cite{fouque03})}.\\
- The deterministic quantity $D_{t,T}$ \rc{given by (\ref{eq:represcalD})} satisfies
$$
D_{t,T}  \leq C  (T-t)^{H+\frac{3}{2}} ,
$$
for any  $t \in [0,T]$ and for some constant $C$  \rc{(see Lemma \ref{lem:2} and below in Appendix \ref{sec:tech})}.\\
- The deterministic quantity $\theta_{t,T}$ \rc{defined by (\ref{def:thetatT})} satisfies
$$
\theta_{t,T}  \leq C  (T-t)^{H+\frac{1}{2}} ,
$$
for any  $t \in [0,T]$ and for some constant $C$ \rc{(substitute (\ref{eq:asyK1}-\ref{eq:asyK2}) into (\ref{def:thetatT}))}.\\
- The random component $\phi_t$ defined by (\ref{def:phit}) satisfies
$$
\EE [ |\phi_t|^p ]^{\frac{1}{p}} \leq C_p (T-t) ,
$$
for any  $t \in [0,T]$ and for some constant $C_p$ for any $p>0$  \rc{(apply Lemma \ref{lem:3} in Appendix \ref{sec:tech} and use the fact that $\phi_t$ is Gaussian)}.\\
- The random process $Z_t^H$ satisfies
$$
\EE [ |Z_t^H|^p ]^{\frac{1}{p}}  \leq C_p   ,
$$
for any  $t \in [0,T]$ and for some constant $C_p$ for any $p>0$ \rc{(use the fact that $Z_t^H$ is Gaussian, stationary, with mean zero and variance $\sigma_{\rm ou}^2$)}.\\
As a consequence, the deterministic function $Q^{(1)}_t(x)$ satisfies
$$
| \partial_x^k Q^{(1)}_t (x) | \leq C \Big( (T-t)^{H+\frac{3}{2}} + (T-t)^{H+\frac{1}{2}-\frac{k}{2}}  \Big)\big(1 +x^3\big) ,
$$
for any $1\leq k \leq 2$, $t \in [0,T]$, $x\in (0,\infty)$, and for some constant $C$.\\
Using  (\ref{def:R4}) 
and the previous estimates we find that, 
for any $p>0$, there exists a constant $C_p$ such that 
\begin{eqnarray*}
\EE [ |R_{t,T}|^p ]^{\frac{1}{p}} 
&\leq& C_p \delta^2 \int_t^T (T-s)^{-\frac{1}{2}} +  (T-s)^{-\frac{1}{2}} +  (T-s)^{H-\frac{1}{2}}  +  (T-s)^{H-\frac{1}{2}} ds  \\
&\leq& C_p \delta^2 \big( (T-t)^{\frac{1}{2}} + (T-t)^{H+\frac{1}{2}}\big) ,
\end{eqnarray*}
for any $\delta \in (0,1)$ and $t \in [0,T]$,
which shows the desired result for $R_{t,T}$.\\
Moreover, the local martingales $N^{(j)}_t$ in  (\ref{def:N4}) are continuous square-integrable martingales 
up to time $T$ whose brackets are
\begin{eqnarray*}
d \left<N^{(j)}\right>_t &=&
{\cal N}^{(j)}_t d t ,  \quad j=0,1,2,\\
{\cal N}^{(0)}_t &=&  \big(  \sigma_t \big(x \partial_x Q_t^{(0)}\big) (X_t) \big)^2, \\
{\cal N}^{(1)}_t &=&
 \big( \big( x\partial_x\big(  x^2 \partial_x^2 \big)\big) Q_t^{(0)}(X_t) \sigma_t  \phi_t \big)^2 
\\
&&
+2 \rho \theta_{t,T} 
\big( \big( x\partial_x\big(  x^2 \partial_x^2 \big)\big) Q_t^{(0)}(X_t) \sigma_t \phi_t \big) 
\big(  \big( x^2 \partial_x^2 \big) Q_t^{(0)}(X_t) \big) \\
&&
+
\big(  \big( x^2 \partial_x^2 \big) Q_t^{(0)}(X_t) \big)^2 \theta_{t,T}^2 ,\\
{\cal N}^{(2)}_t &=& \big( \sigma_t \big( x \partial_x Q_t^{(1)}\big) (X_t) \big)^2 ,
\end{eqnarray*}
where the ${\cal N}^{(j)}_t $ are uniformly bounded with respect to $t \in [0,T]$ in $L^p$ for any $p$, which concludes the proof.
\end{proof}

\section{The implied volatility}
\label{sec:implied}

\rcc{We now compute and discuss the implied  volatility  associated with the price
approximation given in Proposition \ref{prop:0}.
This implied volatility  is the volatility that when used in the constant 
volatility Black-Scholes pricing formula gives the same  price as the approximation,
to  the order of the approximation.  }
The implied volatility in the context of the European option 
introduced in the previous section is then given by
\begin{equation}
\label{eq:iv1}
I_t = \bar{\sigma} + \delta    \frac{\phi_t}{T-t}
+ \delta  \rho D_{t,T}\Big[ \frac{\bar{\sigma}}{2(T-t)} + \frac{\log(K/X_t)}{\bar{\sigma} (T-t)^2} \Big]
+O(\delta^2)  .
\end{equation}
The first two terms can be combined and rewritten as (up to terms of order $\delta^2$):
\begin{equation}
\bar{\sigma} + \delta    \frac{\phi_t}{T-t} = \EE\Big[ \frac{1}{T-t} \int_t^T \sigma_s^2 ds |{\cal F}_t \Big]^{\frac{1}{2}} 
+O(\delta^2).
\end{equation}

When $a(T-t)\ll 1$ the implied volatility is random and we have  
(see  Lemma \ref{lem:3}) and Eq. (\ref{eq:behDsmall}) :
\begin{equation}
\label{eq:iv2}
I_t = \bar{\sigma} + \delta  Z_t^H
+ \delta \frac{ \rho }{ \Gamma(H+\frac{5}{2})}
 \Big[ \frac{\bar{\sigma}}{2}(T-t)^{\frac{1}{2}+H}  + \frac{\log(K/X_t)}{\bar{\sigma} (T-t)^{\frac{1}{2}-H}} \Big]
  .
\end{equation}
Note that, for  $H \in (0,1/2)$, the implied volatility blows up at small time-to-maturity~$T-t$.
\rc{Note,  moreover  that the result above is valid in the asymptotic regime
$\delta \ll 1$.  Indeed, for $\bar{\sigma}$ being an order one strictly positive quantity  
 the implied volatility in Eq. (\ref{eq:iv2})  is strictly positive for $\delta$ small enough.    }  

When $a(T-t) \gg 1$, the quantity $D_{t,T}$ is of order $(T-t)^{H+\frac{1}{2}}$ and is 
deterministic (by Lemma \ref{lem:2}), while the fluctuations of $\phi_t$ 
are of order $(T-t)^H$ at most and are therefore negligible (by Lemma \ref{lem:3}).
As a consequence, when  $a(T-t) \gg 1$, we can write the implied volatility as:
\begin{equation}\label{eq:iv3}
I_t = \bar{\sigma}  
+ \delta \frac{ \rho}{a \Gamma(H+\frac{3}{2})}
 \Big[ \frac{\bar{\sigma}}{2 } (T-t)^{H-\frac{1}{2}}+ \frac{\log(K/X_t)}{\bar{\sigma} (T-t)^{\frac{3}{2}-H}} \Big]
  .
\end{equation}
Note that, for  $H \in (1/2,1)$, the implied volatility blows up at large time-to-maturity~$T-t$.
\bc{ We remark that the factors multiplying the square brackets in
Eqs.~(\ref{eq:iv2})  and (\ref{eq:iv3}) are slightly modified in the general case
when $Z_t$ is a general Gaussian process, see Appendix 
\ref{app:gen}.} 

\section{A slow volatility factor}
\label{sec:conn}
%
%
%
%
%
%

We show in this section that the approach developed in this paper can be applied to other stochastic volatility models.
Here we consider the following model
\begin{equation}\label{eq:fSVslow}
\sigma_t = F( Z_{t}^{\delta,H} ) ,
\end{equation}
where $F$ is a smooth, positive-valued function, bounded away from zero, with bounded derivatives, 
and $Z_{t}^{\delta,H} $ is a rescaled fOU process:
\begin{equation}
d Z_{t}^{\delta,H} = \delta^H d W^H_t - \delta a Z_{t}^{\delta,H} dt  ,
\end{equation}
whose natural time scale is $1/\delta$. It has the form
\begin{equation}
Z_{t}^{\delta,H}  = \delta^H \int_{-\infty}^t e^{-\delta a (t-s)} dW^H_s.
\end{equation}
Its moving-average integral representation is
\begin{equation}
Z^{\delta,H}_t = \int_{-\infty}^t {\cal K}^\delta(t-s) dW_s, \quad \quad {\cal K}^\delta(t) = \delta^{\frac{1}{2}} {\cal K}(\delta t),
\end{equation}
where ${\cal K}$ is defined by (\ref{def:calK}). In particular its variance is $\sigma_{\rm ou}$ defined by (\ref{def:sigmaou}),
independently of $\delta$.
This model is therefore characterized by strong but slow fluctuations of the volatility.
If the price of the risky asset follows the stochastic differential equation (\ref{eq:X}),
we get a result similar to Proposition \ref{prop:0}.

\begin{proposition}
\label{prop:1}%
When $\delta$ is small,
denoting $\sigma_0=F( Z_{0}^{\delta,H} )$ and $p_0=F'(Z_{0}^{\delta,H})$,
the option price (\ref{def:optionprice}) is of the form
\begin{equation}
M_t = Q_t(X_t) +O(\delta^{2H}) ,
\end{equation}
where
\begin{equation}
Q_t(x) = Q_t^{(0)}(x) +  \sigma_0 p_0 \phi_t^\delta \big( x^2\partial_x^2 Q_t^{(0)}(x)\big)
+\delta^H \rho p_0 Q_t^{(1)}(x) ,
\end{equation}
$Q_t^{(0)}(x) $ is given by the Black-Scholes formula with constant volatility~$\sigma_0$,
\begin{equation}
\label{eq:bs00}
{\cal L}_{\rm BS} (\sigma_0) Q_t^{(0)}(x) =0,   \quad \quad Q_T^{(0)}(x) = h(x),
\end{equation}
$\phi_t^\delta$ is the random component
\begin{equation}
\label{def:phitdelta}
\phi_t^\delta = \EE \Big[\int_t^T Z_s^{\delta,H}-Z_0^{\delta,H} ds |{\cal F}_t \Big],
\end{equation}
and $Q_t^{(1)}(x) $ is the correction
\begin{equation}
\label{def:Q1tdelta}
Q_t^{(1)}(x) = \sigma_0^2 x \partial_x \big(x^2 \partial_x^2 Q_t^{(0)}(x)  \big) D_{t,T} ,
\end{equation}
with $D_{t,T}$ defined by
\begin{equation}
D_{t,T} = \frac{(T-t)^{H+\frac{3}{2}}}{\Gamma(H+\frac{5}{2})} .
\end{equation}
\end{proposition}
\rc{
   We remark that we indeed can expect  the situation with a slow volatility factor to 
   behave qualitatively as the situation with small  volatility fluctuations in Proposition \ref{prop:0} 
   from the point of view of the effect  the medium roughness. This follows 
   since we have from self-similarity of fractional Brownian motion that in distribution 
   $$
        \delta W_t^H    \stackrel{d}{=}  W_{\delta^{  {1}/{H}  } t}^{H} . 
    $$
    However,  we have
    $$
        \delta Z_t^H\mid_{a=a'}    \stackrel{d}{=}  Z_{\delta^{  {1}/{H}  }t}^{H}\mid_{a=\delta^{1/H} a'} ,  
    $$
   and thus the models (small volatility fluctuations versus slow) 
   differ both in a  strong sense and in distribution. Moreover,
    the models have  different interpretations   from the modeling viewpoint with
    for instance a  different skewness mechanism.  
    Note in particular that  this difference manifests itself
   in that for fixed Hurst coefficient $H$ the magnitude of both the correction and
   the error terms, deriving in particular from the modeling of correlation, 
   have a different scaling in $\delta$ for the two models.  For instance
   for small Hurst exponent $H$ we may expect, for given $\delta$,  the correction 
   (and also the error term) to be relatively larger in  the case of  the slow volatility
   factor.   
The random correction $\phi_t^\delta$ is of order $\delta^H$.
More exactly it is a zero-mean Gaussian random variable with variance
\begin{eqnarray}
\nonumber
\EE \big[ (\phi_{t}^\delta)^2\big] &=& \frac{\delta^{2H} T^{2+2H}}{\Gamma(H+\frac{3}{2})^2} 
\int_0^\infty \Big[ \big(1-\frac{t}{T} +v\big)^{H+\frac{1}{2}}-v^{H+\frac{1}{2}}  \\
&& - \big(1-\frac{t}{T} \big) \big(H+\frac{1}{2}\big) \big(v-\frac{t}{T} \big)_+^{H-\frac{1}{2}}
\Big]^2 dv +O(\delta^{2H+1}) ,
\label{eq:varphidelta}
\end{eqnarray}
for $t \in [0,T]$.}

\begin{proof}
We note that
$$
\sigma_t = \sigma_0 + p_0 (Z_t^{\delta,H}-Z_0^{\delta,H})+ g^\delta_t  ,
$$
where $g^\delta_t=F(Z_t^{\delta,H})-F(Z_0^{\delta,H})-F'(Z_0^{\delta,H}) (Z_t^{\delta,H}-Z_0^{\delta,H})$ and therefore
$$
| g^\delta_t |\leq \frac{1}{2} \| F''\|_\infty ( Z_t^{\delta,H}-Z_0^{\delta,H})^2 .
$$
We have
$$
\EE \big[ ( Z_t^{\delta,H}-Z_0^{\delta,H})^2 \big] = \int_0^{\delta t} {\cal K}(s)^2ds +\int_0^\infty \big[ {\cal K}(\delta t+s) -{\cal K}(s)\big]^2 ds ,
$$
which is of order $\delta^{2H}$:
$$
\EE \big[ ( Z_t^{\delta,H}-Z_0^{\delta,H})^2 \big] = \sigma_H^2 (\delta t)^{2H}
+o(\delta^{2H}) .
$$
Therefore $g^\delta_t$ is bounded in $L^p$ for any $p$ by a quantity of order $\delta^{2H}$.
We can then follow the same proof as the one of Proposition \ref{prop:0}.
The term 
$$
D_{t,T}^\delta = \int_0^\tau (\tau-u) {\cal K}^\delta(u) du ,
$$
is given by
$$
D_{t,T}^\delta = \delta^H \frac{(T-t)^{H+\frac{3}{2}}}{\Gamma(H+\frac{5}{2})} + O(\delta^{2H}).
$$
The variance of the correction 
$\phi_t^\delta$ is
$$
\EE \big[ (\phi_{t}^\delta)^2\big] = \int_0^t \Big( \int_t^T {\cal K}^\delta(s-u) ds \Big)^2 du 
+
\int_{-\infty}^0 \Big( \int_t^T {\cal K}^\delta(s-u) -{\cal K}^\delta(-u) ds \Big)^2 du ,
$$
which in turn gives (\ref{eq:varphidelta}).
\end{proof}

Proceeding as in the case of the small-amplitude stochastic volatiliy model,
we find that the implied volatility in the context of the European option is given by
\begin{equation}
I_t = \sigma_0 +  p_0 \frac{\phi_t^\delta}{T-t}
+  \delta^H  \frac{\rho p_0 }{\Gamma(H+\frac{5}{2})}\Big[ 
\frac{\sigma_0}{2} (T-t)^{H+\frac{1}{2}}+ \frac{\log(K/X_t)}{\sigma_0 (T-t)^{\frac{1}{2}-H}} \Big]
+O(\delta^{2H})  .
\end{equation}
The first two terms can be combined and rewritten as (up to terms of order $\delta^{2H}$):
\begin{equation}
\sigma_0 + p_0 \frac{\phi_t^\delta}{T-t}= \EE\Big[ \frac{1}{T-t} \int_t^T \sigma_s^2 ds |{\cal F}_t \Big]^{\frac{1}{2}} 
+O(\delta^{2H}).
\end{equation}

\section{Conclusion}
\label{sec:concl}%
We have presented an analysis of the European option price when the volatility is stochastic 
and has correlations that decay 
as a fractional power of the time offset.   
The stochastic volatility model is defined in terms 
of a fractional Ornstein Uhlenbeck process with Hurst exponent $H$
and the analysis is carried out when the typical amplitude of the volatility fluctuations is relatively small. 
Two situations are differentiated. First the situation
when $H \in (0,1/2)$ which corresponds to a ``short-range'' 
dependent process that is 
rough on short scales with correlations that decay very rapidly, faster
than linear decay, at the origin. 
Second the situation when $H \in (1/2,1)$ so that   
the correlations decay relatively 
slowly at large scales and then the volatility correlations are not
integrable.  We use a  martingale method approach to derive
a general expression for the Black-Scholes price covering the two
cases. 
In the short-range case the  rough  behavior 
on short scales gives rise to an implied volatility that 
diverges as the time to maturity goes to zero. 
In the long-range case the slow decay in the correlations 
gives a term structure of the implied volatility that diverges
as time to maturity goes to infinity. 
The main result  we have presented is specific in the sense that a particular stochastic volatility model
has been addressed, 
however, as we illustrate the framework can be adapted to related  models as long as 
some central covariance terms  can be computed.
We illustrate this by considering a model with slow, but order one, volatility fluctuations
and derive the associated fractional implied volatility term structure. 

\appendix

\section{Technical lemmas}
\label{sec:tech}%
In this appendix we state and prove a few technical lemmas related to 
some central quantities of interest that are used in the derivation of the price in Sections \ref{sec:option}
and \ref{sec:implied}.

The martingale $\psi_t$ is defined 
for any $t \in [0,T]$ by (\ref{def:Kt}).
It is used in the proof of Proposition  \ref{prop:0} and it has the following properties.
\begin{lemma}
\label{lem:1}%
$(\psi_t)_{t\in [0,T]}$ is a Gaussian square-integrable martingale and
\begin{equation}
d \left< \psi, W\right>_t = \Big(    \int_0^{T-t} {\cal K}(s) ds \Big) dt   ,
\quad \quad
d \left< \psi \right>_t =    \Big( \int_0^{T-t} {\cal K}(s) ds   \Big)^2 dt.
\end{equation}
\end{lemma}

\begin{proof}
For $t\leq s$, the conditional distribution of $Z_s^H$ given ${\cal F}_t$ is Gaussian with mean
\begin{equation}\label{eq:mean}
\EE \big[  Z_s^H |{\cal F}_t \big] = \int_{-\infty}^t {\cal K}(s-u) dW_u  ,
\end{equation}
and deterministic variance given by
$$
 {\rm Var} \big( Z_s^H |{\cal F}_t\big) = 
   \int_0^{s-t} {\cal K}(u)^2 du .
$$
Therefore we have
\begin{eqnarray*}
\psi_t &=& \int_0^t    Z_s^H  ds + \int_t^T \EE \big[  Z_s^H  |{\cal F}_t \big]  ds \\
&=& 
\int_0^t ds \int_{-\infty}^s {\cal K}(s-u) dW_u +
\int_t^T dt \int_{-\infty}^t {\cal K}(s-u) dW_u \\
&=&
   \int_{-\infty}^0 \Big[  \int_0^T {\cal K}(s-u) dt \Big] dW_u +
 \int_0^t  \Big[ \int_u^T  {\cal K}(s-u) dt \Big]  dW_u.
\end{eqnarray*}
This gives
\begin{eqnarray*}
d \left< \psi, W\right>_t
=   \Big( \int_t^T{\cal K}(s-t) ds \Big) dt  ,
\quad \quad 
d \left< \psi \right>_t
= \Big(  \int_t^T{\cal K}(s-t) ds \Big)^2 dt  ,
\end{eqnarray*}
as stated in the Lemma.
\end{proof}

We define the deterministic component 
\begin{equation}
\label{def:DtT}
D_{t,T} = \left< \psi, W\right>_T -  \left< \psi, W\right>_t  ,
\end{equation} 
that appears in Equation (\ref{def:Q1t}).
It has the following properties.

\begin{lemma}
\label{lem:2}
$D_{t,T} $ is a deterministic function of $T-t$ and it is given by
\begin{equation}
 D_{t,T} =   {\cal D}(T-t), \quad \quad {\cal D}(\tau) =  \int_0^\tau (\tau-u){\cal K}(u) du .
\end{equation}
\end{lemma}
The function ${\cal D}$ can be written as (\ref{eq:represcalD})
and it has the following behavior:\\
For $a\tau \ll 1$, 
\begin{equation}
\label{eq:behDsmall}
{\cal D}(\tau) = \frac{1}{ \Gamma(H+\frac{5}{2}) a^{H+\frac{3}{2}} }
 \Big( (a\tau)^{H+\frac{3}{2}} 
+o\big( (a\tau)^{H+\frac{3}{2}}  \big) \Big).
\end{equation}
For $a\tau \gg 1$, 
\begin{equation}
\label{eq:behDlarge}
{\cal D}(\tau) = \frac{1}{ \Gamma(H+\frac{3}{2}) a^{H+\frac{3}{2}}}
\Big( (a\tau)^{H+\frac{1}{2}}  +o\big( (a\tau)^{H+\frac{1}{2}}  \big) \Big).
\end{equation}

Finally, we consider the random process $\phi_t$ defined by (\ref{def:phit}).
\begin{lemma}
\label{lem:3}
\begin{enumerate}
\item
 $\phi_t$ is a zero-mean Gaussian process with variance
 \begin{eqnarray*}
 {\rm Var}(\phi_t) = 
 \int_0^\infty \Big( \int_0^{T-t} {\cal K}(s+u) ds  \Big)^2 du  .
 \end{eqnarray*}
 \item
 There exists a constant $C$ (that depends on $H$)
such that the variance of $\phi_t$ can be bounded by
\begin{equation}
{\rm Var}(\phi_t) \leq C \, (T-t)^{2H} \wedge (T-t)^2.
\end{equation}
\item
$\phi_t$ is approximately equal to $(T-t) Z_t^H$ for small $T-t$:
\begin{equation}
\EE \Big[ \Big(\frac{\phi_t}{T-t} -Z_t^H \Big)^2 \Big]\stackrel{T-t \to 0}{\longrightarrow} 0 .
\end{equation}
\end{enumerate}
\end{lemma}

\begin{proof}
We can express the variance of $\phi_t$ as:
\begin{eqnarray*}
{\rm Var}(\phi_t) =  \int_0^{T-t} ds \int_0^{T-t} ds' 
{\rm Cov}\big( \EE \big[ Z_s^H|{\cal F}_0\big] ,\EE \big[ Z_{s'}^H|{\cal F}_0\big] \big) ,
\end{eqnarray*}
which gives the first item since 
$$
{\rm Cov}\big( \EE \big[ Z_s^H|{\cal F}_0\big] ,\EE \big[ Z_{s'}^H|{\cal F}_0\big] \big) =
\int_{-\infty}^0 {\cal K}(s-u) {\cal K}(s'-u) du .
$$
Furthermore
\begin{eqnarray*}
{\rm Var}(\phi_t) &\leq &  \Big( \int_0^{T-t} 
{\rm Var}\big( \EE \big[ Z_s^H|{\cal F}_0\big]\big)^{1/2} ds
\Big)^2 \\
&\leq & \Big(  \int_0^{T-t}    
   \Big(\int_s^{\infty} {\cal K}(u)^2 du\Big)^{1/2} ds
 \Big)^2\\
&\leq &  C \, (T-t)^{2H} \wedge (T-t)^2,
\end{eqnarray*}
which gives the second item of the lemma.

Similarly, we have
\begin{eqnarray*}
\EE \Big[\Big(\frac{\phi_t}{T-t}-Z_t^H \Big)^2 \Big] &\leq &
\Big( \frac{1}{T-t} \int_0^{T-t} ds 
{\rm Var}\big( \EE \big[ Z_s^H|{\cal F}_0\big] - Z_0^H \big)^{1/2}
\Big)^2, 
\end{eqnarray*}
and
\begin{eqnarray*}
 \EE \big[ Z_s^H|{\cal F}_0\big] - Z_0^H
 =
  \int_{-\infty}^0  \big( {\cal K}(s-u) - {\cal K}(-u)\big) dW_u    ,
\end{eqnarray*}
so that 
\begin{eqnarray*}
\EE \Big[\Big(\frac{\phi_t}{T-t}-Z_t^H \Big)^2 \Big] &\leq &
\Big( \frac{1}{T-t} \int_0^{T-t} \Big[
 \int_0^\infty \big({\cal K}(s+v)-{\cal K}(v) \big)^2 dv \Big] ds \Big)^2  .
\end{eqnarray*}
As $s \to 0$, we have $ \int_0^\infty \big({\cal K}(s+v)-{\cal K}(v) \big)^2 dv \to 0$
by Lebesgue's dominated convergence theorem
(remember ${\cal K}\in L^2$),
which gives the third item.
\end{proof}

\section{\bc{Extension to a general stochastic volatility model}}\label{app:gen} 
In the paper, we model the volatility as a bounded function of a fOU process.
In fact it is straightforward to extend all the results to a volatility model that is a bounded
function of a stationary Gaussian process whose correlation properties are qualitatively similar
as the ones of a fOU process. 
In this appendix we consider the situation when the volatility is
\begin{equation}
\label{eq:fSV:app}
\sigma_t =\bar{\sigma}+F(\delta Z_t) ,
\end{equation}
for $Z_t$ a stationary Gaussian process with mean zero of the form
\begin{equation}
Z_t = \int_{-\infty}^t {\cal K}(t-s) dW_s   ,
\end{equation}
where  $W_t$ is a standard Brownian motion and
${\cal K} \in L^2(0,\infty)$ is a general kernel instead of the specific kernel (\ref{def:calK})
corresponding to a fOU.
Then the Gaussian process $Z_t$ has mean zero, variance
\begin{equation}
\sigma^2_Z = \int_0^\infty {\cal K}^2(u) du  ,
\end{equation}
and covariance
$$
\EE [ Z_t Z_{t+s}] = \int_0^\infty {\cal K}(u) {\cal K}(u+s) du .
$$
As  before (above Proposition \ref{prop:0}), 
the function $F$ is assumed to be one-to-one, 
smooth, bounded from below by a constant larger than $-\bar{\sigma}$,
with bounded derivatives,
and such that $F(0)=0$ and $F'(0)=1$. 
Proposition \ref{prop:0} then holds true, with the function ${\cal D}$ defined by 
\begin{equation}
\label{def:calDB}
{\cal D}(\tau ) = \int_0^\tau (\tau-u) {\cal K}(u) du ,
\end{equation}
and the implied volatility in the context of the European option  is still given by (\ref{eq:iv1}) with $D_{t,T}={\cal D}(T-t)$.
The behavior of the function ${\cal D}$ is determined by the one of the kernel ${\cal K}$ and we
consider in more detail two cases corresponding respectively to long-  and short-range correlations:
\begin{enumerate}
\item
There exists $c_Z  \neq 0$ such that 
\begin{equation}
\label{eq:lrLapp}
{\cal K}(t) = c_Z t^{H-\frac{3}{2}} \big(1+o(1)\big) \mbox{ as }t \to \infty. 
\end{equation}
If $H \in (1/2,1)$ this implies that ${\cal K}$
is not integrable at infinity and, as we will see below (see Lemma \ref{lem:1:app}),
 the covariance function of $Z_t$ has a tail behavior similar 
to that of a fOU at infinity. In other words,  $Z_t$ possesses long-range correlation properties, and 
the implied volatility has the same form (\ref{eq:iv3}) as in the case of a fOU with Hurst index $H$,
with $c_Z \Gamma(H-1/2) / \Gamma(H+3/2) $ instead of $1/[ a \Gamma(H+3/2)]$.\\
\item There exists $d_Z  \neq 0$ such that 
\begin{equation}
\label{eq:srLapp}
{\cal K}(t) = d_Z t^{H-\frac{1}{2}} \big(1+o(1)\big) \mbox{ as } t \to 0.
\end{equation}
If $H \in (0,1/2)$ this implies that ${\cal K}$
is singular at zero and, as we will see below  (see Lemma \ref{lem:2:app}), 
the covariance function of $Z_t$ has a behavior similar to  that of a fOU at zero. 
 In other words,  $Z_t$ possesses  short-range correlation properties, and 
 the implied volatility has the same form (\ref{eq:iv2}) as in the case of a fOU with Hurst index $H$,
 however, with $d_Z \Gamma(H+1/2)/  \Gamma(H+5/2)$ replacing  $1/  \Gamma(H+5/2)$.
\end{enumerate}

\begin{lemma}
\label{lem:1:app}%
We assume (\ref{eq:lrLapp}).
\begin{enumerate}
\item
If $H \in (1/2,1)$, then  the covariance function of $Z_t$ satisfies 
\begin{equation}
\label{lem:1:app:eq1}
\EE [ Z_t Z_{t+s} ] = k_Z s^{2H-2} \big(1+o(1)\big)  \mbox{ as } s \to \infty ,
\end{equation}
with
\begin{equation}
\label{def:kZ}
k_Z = c_Z^2  \frac{\Gamma(2-2H)\Gamma(H-\frac{1}{2})}{\Gamma(\frac{3}{2}-H)}
=
c_Z^2 \frac{\Gamma(H-\frac{1}{2})^2 }{2\sin (\pi H) \Gamma(2H-1)} .
\end{equation}
\item
If $H \in (1/2,1)$, then the function ${\cal D}(\tau)$ defined by (\ref{def:calDB}) satisfies
\begin{equation}
\label{lem:1:app:eq2}
{\cal D} (\tau) = c_Z \frac{\Gamma(H-\frac{1}{2})}{\Gamma(H+\frac{3}{2})} \tau^{H+\frac{1}{2}} \big(1+o(1)\big)  \mbox{ as } \tau \to \infty. 
\end{equation}
\end{enumerate}
\end{lemma}
If $Z_t$ is the fOU process (\ref{eq:Zd}), we have $c_Z = 1/[a\Gamma(H-1/2)]$.
In this case,
we can check that 
$k_Z = a^{-2} / [2 \sin (\pi H) \Gamma(2H-1) ] = \sigma_{\rm ou}^2 a^{2H-2} / \Gamma(2H-1) $,
which confirms that (\ref{lem:1:app:eq1}-\ref{def:kZ}) give (\ref{eq:corZG3}),
while (\ref{lem:1:app:eq2}) gives (\ref{eq:behDlarge}).

\begin{proof}
We denote
$$
{\cal C}(s) = \EE [ Z_t Z_{t+s}] = \int_0^\infty {\cal K}(u) {\cal K}(u+s) du
\mbox{ and }
\tilde{\cal C}(s) = c_Z^2  \int_0^\infty u^{H-\frac{3}{2}} (u+s)^{H-\frac{3}{2}} du .
$$
We can check that $\tilde{\cal C}(s)= k_Z s^{2H-2}$ with
$$
k_Z = c_Z^2 \int_0^\infty u^{H-\frac{3}{2}} (1+u)^{H-\frac{3}{2}} du 
= c_Z^2 \frac{\Gamma(2-2H)\Gamma(H-\frac{1}{2})}{\Gamma(\frac{3}{2}-H)}. 
$$
We now show that ${\cal C}(s)-\tilde{\cal C}(s)$ goes to zero as $s \to \infty$ faster than $s^{2H-2}$.
Let $\eps\in (0,1)$. There exists $S^\eps$ such that $|{\cal K}(t)  t^{-H+3/2} -c_Z |\leq \eps$ for any $t \geq S^\eps$.
We have for any $s \geq S^\eps$:
\begin{eqnarray*}
s^{2-2H}\big|{\cal C}(s)-\tilde{\cal C}(s)\big|
&\leq&
s^{2-2H} \int_0^{S^\eps} \big|
 {\cal K}(u) {\cal K}(u+s) - 
c_Z^2  u^{H-\frac{3}{2}} (u+s)^{H-\frac{3}{2}} \big| du
\\
&&
+s^{2-2H} \int_{S^\eps}^\infty
|{\cal K}(u)| |{\cal K}(u+s) - c_Z (u+s)^{H-\frac{3}{2}}| du  \\
&&
+s^{2-2H} \int_{S^\eps}^\infty
|c_Z| (u+s)^{H-\frac{3}{2}} |{\cal K}(u) - c_Z u^{H-\frac{3}{2}}| du \\
&\leq&
s^{2-2H} \int_0^{S^\eps} \big|
 {\cal K}(u) {\cal K}(u+s) - 
c_Z^2  u^{H-\frac{3}{2}} (u+s)^{H-\frac{3}{2}} \big| du
\\
&&
+s^{2-2H} \eps( 2|c_Z| +\eps) \int_0^\infty (u+s)^{H-\frac{3}{2}} u^{H-\frac{3}{2}} du  .
\end{eqnarray*}
As $s \to \infty$ the first term of the right-hand side goes to zero by Lebesgue dominated convergence theorem
because $(2-2H)+(H-3/2)<0$.
This gives
\begin{eqnarray*}
\limsup_{s \to \infty}
s^{2-2H}\big|{\cal C}(s)-\tilde{\cal C}(s)\big|
&\leq&
 \eps( 2|c_Z| +\eps) \int_0^\infty (u+1)^{H-\frac{3}{2}} u^{H-\frac{3}{2}} du.
\end{eqnarray*}
Since this holds true for any $\eps \in(0,1)$, this proves (\ref{lem:1:app:eq1}).

We denote
$$
\tilde{\cal D}(\tau) = \int_0^\tau (\tau-u) c_Z u^{H-\frac{3}{2}} du ,
$$
which is given by 
$$
\tilde{\cal D}(\tau) = \frac{c_Z }{H^2-\frac{1}{4}} \tau^{H+\frac{1}{2}}  
= c_Z\frac{ \Gamma(H-\frac{1}{2})}{\Gamma(H+\frac{3}{2})} \tau^{H+\frac{1}{2}}  .
$$
Let $\eps \in (0,1)$.  As mentioned above, there exists $S^\eps$ such that $|{\cal K}(t)  t^{-H+3/2} -c_Z |\leq \eps$ for any $t \geq S^\eps$.
We have then for any $\tau \geq S^\eps$:
\begin{eqnarray*}
\tau^{-H-\frac{1}{2}} \big| {\cal D}(\tau) - \tilde{\cal D}(\tau) \big|
& \leq &
\tau^{-H -\frac{1}{2}} \int_0^{S^\eps} (\tau-u) \big| {\cal K}(u) -c_Z u^{H-\frac{3}{2}}\big| du
\\
&& +
\tau^{-H -\frac{1}{2}} \eps \int_{S^\eps}^\tau (\tau-u) u^{H-\frac{3}{2}} du  .
\end{eqnarray*}
As $\tau \to \infty$ the first term of the right-hand side goes to zero by Lebesgue dominated convergence theorem
because $-(H+1/2)+1<0$.
This gives
\begin{eqnarray*}
\limsup_{\tau \to \infty}
\tau^{-H-\frac{1}{2}} \big| {\cal D}(\tau) - \tilde{\cal D}(\tau) \big|
& \leq &
\eps
\int_{0}^1 (1-u) u^{H-\frac{3}{2}} du .
\end{eqnarray*}
Since this holds true for any $\eps \in(0,1) $, this proves (\ref{lem:1:app:eq2}).
\end{proof}

\begin{lemma}
\label{lem:2:app}%
We assume (\ref{eq:srLapp}).
\begin{enumerate}
\item
If $H \in (0,1/2)$ and if ${\cal K}$ satisfies the two technical conditions:
\begin{enumerate}
\item[(CB.2.1)] ${\cal K}$
is integrable and Lipschitz on $(1,\infty)$.
\item[(CB.2.2)]
There exist functions $k_1(t)$ and $k_2(s)$ such that for all $t,s\in (0,1)$ we have 
$| \tilde{\cal K}(t+s)-\tilde{\cal K}(t)|\leq k_1(t) k_2(s)
$, where $\tilde{\cal K}(t) = {\cal K}(t) - d_Z t^{H-1/2}$, $k_1 \in L^2(0,1)$, and $\lim_{s \to 0} s^{-H} k_2(s) =0$.
\end{enumerate}
Then  the covariance function of $Z_t$ satisfies 
\begin{equation}
\label{lem:2:app:eq1}
\EE [ Z_t Z_{t+s} ] = \sigma^2_Z - q_Z s^{2H} + o(s^{2H} )  \mbox{ as } s \to 0 ,
\end{equation}
with 
\begin{eqnarray}
\label{def:qZ}
q_Z &=&  \frac{d_Z^2}{2}  \frac{\Gamma(H+\frac{1}{2})^2}{\Gamma(2H+1)\sin (\pi H)}   , \\
\sigma_Z^2 &=&   \int_{0}^\infty   {\cal K}^2(u)  du. \label{def:sZ}
\end{eqnarray}  
\item
For any $H \in (0,1)$,
the function ${\cal D}(\tau)$ defined by (\ref{def:calDB}) satisfies
\begin{equation}
\label{lem:2:app:eq2}
{\cal D} (\tau) = d_Z \frac{ \Gamma(H+\frac{1}{2})}{\Gamma(H+\frac{5}{2})} \tau^{H+\frac{3}{2}} 
\big( 1+o(1) \big)  \mbox{ as } \tau \to 0. 
\end{equation}
\end{enumerate}
\end{lemma}
The condition (CB.2.1) gives some control of  ${\cal K}$ away from the origin, and  
this specific condition can be relaxed.  
The necessary condition is that  
\[s^{-2H}  \int_{1}^\infty \big( {\cal K}(u+s)-{\cal K}(u)\big)^2 du\]
goes to zero as $s \to 0$ (see the proof below). 

The condition (CB.2.2) means that the remainder $\tilde{\cal K}(t) $ should be small enough near the origin.
A sufficient condition for (CB.2.2) is that 
$\tilde{\cal K}$ is $\alpha$-H\"older continuous over $(0,2)$ for some $\alpha>H$. 
Then (CB.2.2) is fulfilled 
with $k_1(t)=c$ and $k_2(s) = \tilde{k}_\alpha s^\alpha$, for some constant $c$.

If $Z_t$ is the fOU process (\ref{eq:Zd}),  then ${\cal K}$ is integrable and
Lipschitz  on $(1,\infty)$ and we have $d_Z=1/\Gamma(H+1/2)$ and
$\tilde{\cal K}(t) = -[a/\Gamma(H+1/2)] \int_0^t (t-s)^{H-1/2} e^{-as} ds$,  
which is $(H+1/2)$-H\"older continuous over $(0,2)$:
$|\tilde{\cal K}(t+s)-\tilde{\cal K}(t)|\leq[2a/\Gamma(H+3/2)] s^{H+1/2}$.
In this case we can check that 
$q_Z = 1/[2\Gamma(2H+1)\sin (\pi H)]= \sigma_{\rm ou}^2  a^{2H}/\Gamma(2H+1)$,
moreover we have then  $\sigma_Z^2=  \sigma_{\rm ou}^2$, 
which confirms that (\ref{lem:2:app:eq1}-\ref{def:sZ}) give (\ref{eq:corZG2}),
while (\ref{lem:2:app:eq2}) gives (\ref{eq:behDsmall}).

\begin{proof}
We can write
$$
\EE [ Z_t Z_{t+s}] =   \sigma^2_Z  -{\cal Q}(s) , \quad \quad  {\cal Q}(s)= \frac{1}{2}  \EE \big[ (Z_{t+s}-Z_t)^2 \big]  .
$$
We have $ {\cal Q}(s) ={\cal Q}_1(s) +{\cal Q}_2(s)$ with
$$
 {\cal Q}_1(s) =  \frac{1}{2} \int_0^\infty \big( {\cal K}(u+s)-{\cal K}(u)\big)^2 du,
 \quad \quad
 {\cal Q}_2(s)= \frac{1}{2}  \int_0^s {\cal K}(u)^2 du.
$$
The idea is to approximate these two functions by their versions when $d_Z t^{H-1/2}$ replaces  ${\cal K}(t)$.
We denote
$$
\tilde{\cal Q}_1(s) = \frac{d_Z^2}{2} \int_0^\infty \big(  (u+s)^{H-\frac{1}{2}} - u^{H-\frac{1}{2}} \big)^2 du ,\quad\quad
\tilde{\cal Q}_2(s) =  \frac{d_Z^2}{2} \int_0^s u^{2H-1} du .
$$
We can check that $\tilde{\cal Q}_1(s)+\tilde{\cal Q}_2(s)= q_Z s^{2H}$ with
$$
q_Z = \frac{d_Z^2}{2} \int_0^\infty \big( u^{H-\frac{1}{2}} - (1+u)^{H-\frac{1}{2}} \big)^2 du 
+\frac{d_Z^2}{2} \int_0^1  u^{2H-1} du
=\frac{d_Z^2}{2}  \frac{\Gamma(H+\frac{1}{2})^2}{\Gamma(2H+1)\sin (\pi H)} .
$$
We now show that ${\cal Q}_1(s)-\tilde{\cal Q}_1(s)$ goes to zero as $s \to 0$ faster than $s^{2H}$.
We have 
\begin{eqnarray*}
&& 2 s^{-2H}\big|{\cal Q}_1(s)-\tilde{\cal Q}_1(s)\big| \\
&&\leq 
s^{-2H} \Big| \int_0^{1}  \big( {\cal K}(u+s)-{\cal K}(u)\big)^2
- d_Z^2 \big(  (u+s)^{H-\frac{1}{2}} - u^{H-\frac{1}{2}} \big)^2 du\Big| \\
&& \quad
+ s^{-2H}  d_Z^2 \int_{1}^\infty \big(  (u+s)^{H-\frac{1}{2}} - u^{H-\frac{1}{2}} \big)^2 du
 + s^{-2H}  \int_{1}^\infty \big( {\cal K}(u+s)-{\cal K}(u)\big)^2 du\\
&& \leq 
2 |d_Z| s^{-2H} \Big[ \int_0^{1}   \big(  (u+s)^{H-\frac{1}{2}} - u^{H-\frac{1}{2}} \big)^2  du \Big]^{1/2}
\Big[ \int_0^{1}   \big(  \tilde{\cal K}(u+s) - \tilde{\cal K} (u) \big)^2  du \Big]^{1/2}\\
&& \quad + s^{-2H}   \int_0^{1}   \big(  \tilde{\cal K}(u+s) - \tilde{\cal K} (u) \big)^2  du 
 + s^{-2H}  d_Z^2 \int_{1}^\infty \big(  (u+s)^{H-\frac{1}{2}} - u^{H-\frac{1}{2}} \big)^2 du \\
&& \quad
 + s^{-2H}  \int_{1}^\infty \big| {\cal K}(u+s)- {\cal K}(u)\big| \big( |{\cal K}(u+s)| + |{\cal K}(u)| \big)  du\\
&&\leq 
2 |d_Z|  \Big[ \int_0^{\infty}   \big(  (u+1)^{H-\frac{1}{2}} - u^{H-\frac{1}{2}} \big)^2  du \Big]^{1/2}
\Big[ s^{-2H} \int_0^{1}   \big(  \tilde{\cal K}(u+s) - \tilde{\cal K} (u) \big)^2  du \Big]^{1/2}\\
&& \quad + s^{-2H}   \int_0^{1}   \big(  \tilde{\cal K}(u+s) - \tilde{\cal K} (u) \big)^2  du 
+ d_Z^2  \int_{1/s}^\infty \big(  (u+1)^{H-\frac{1}{2}} - u^{H-\frac{1}{2}} \big)^2 du  \\
&& \quad+ 2 s^{1-2H} L_{\cal K}   \int_0^\infty | {\cal K}(u)|du ,
\end{eqnarray*}
where $L_{\cal K}$ is the Lipschitz constant of ${\cal K}$ over $(1,\infty)$.
As $s \to 0$ 
the third term of the right-hand side goes to zero because the integral is convergent and
the fourth term goes to zero because $1-2H>0$. The first and second terms go to zero because
$$
 s^{-2H}   \int_0^{1}   \big(  \tilde{\cal K}(u+s) - \tilde{\cal K} (u) \big)^2  du 
 \leq
s^{-2H} k_2(s)^2  \int_0^1 k_1(u)^2 du ,
$$
$k_1 \in L^2(0,1)$, and $s^{-H} k_2(s) \to 0$ as $s\to 0$.
Therefore
$$
\lim_{s \to 0}
s^{-2H}\big|{\cal Q}_1(s)-\tilde{\cal Q}_1(s)\big|
=0 .
$$
We now show that ${\cal Q}_2(s)-\tilde{\cal Q}_2(s)$ goes to zero as $s \to 0$ faster than $s^{2H}$.
Let $\eps \in (0,1)$. There exists $S^\eps$ such that $|{\cal K}(t)   t^{-H+1/2} -d_Z |\leq \eps$ for any $t \leq S^\eps$.
We have for any $s \leq S^\eps$:
\begin{eqnarray*}
2   s^{-2H } \big| {\cal Q}_2(s) - \tilde{\cal Q}_2(s) \big|
& \leq &
s^{-2H } \int_0^s \big|{\cal K}(u)-d_H u^{H-\frac{1}{2}}\big|
 \big( |{\cal K}(u)|+ |d_Z| u^{H-\frac{1}{2}}\big) du\\
& \leq & 
s^{-2H} \eps (2 |d_Z| +\eps) \int_0^s u^{2H-1} du \leq \frac{(2 |d_Z| +\eps)\eps}{2H} .
\end{eqnarray*}
Since this holds true for any $\eps \in (0,1)$, we have
$$
\lim_{s \to 0}
s^{-2H}\big|{\cal Q}_2(s)-\tilde{\cal Q}_2(s)\big|
=0 ,
$$
which completes the proof of (\ref{lem:2:app:eq1}).

We denote
$$
\tilde{\cal D}(\tau) =  d_Z \int_0^\tau (\tau-u) u^{H-\frac{1}{2}} du ,
$$
which is given by 
$$
\tilde{\cal D}(\tau) = \frac{d_Z }{(H+\frac{1}{2})(H+\frac{3}{2})} \tau^{H+\frac{3}{2}}  = 
d_Z \frac{\Gamma(H+\frac{1}{2})}{\Gamma(H+\frac{5}{2})} \tau^{H+\frac{3}{2}}  .
$$
Let $\eps \in (0,1)$. There exists $S^\eps$ such that $|{\cal K}(t)   t^{-H+1/2} -d_Z |\leq \eps$ for any $t \leq S^\eps$.
We have for any $\tau \leq S^\eps$:
\begin{eqnarray*}
\tau^{-H-\frac{3}{2}} \big| {\cal D}(\tau) - \tilde{\cal D}(\tau) \big|
& \leq &
\tau^{-H -\frac{3}{2}} \eps \int_0^\tau (\tau-u) u^{H-\frac{1}{2}} du .
\end{eqnarray*}
This gives
\begin{eqnarray*}
\limsup_{\tau \to 0}
\tau^{-H-\frac{3}{2}} \big| {\cal D}(\tau) - \tilde{\cal D}(\tau) \big|
& \leq &
\eps\int_0^1(1-u) u^{H-\frac{1}{2}} du .
\end{eqnarray*}
Since this holds true for any $\eps \in (0,1)$, this proves (\ref{lem:2:app:eq2}).
\end{proof}


\begin{thebibliography}{99}

\bibitem{ait}
Y.  A\"it-Sahalia, J. Fan, and  J.  Li,
Testing for jumps in a discretely observed process,
Ann. Statist. {\bf 37} (2009), pp.~184--222. 


\bibitem{ait2}
Y.  A\"it-Sahalia and J. Jacod,
The leverage effect puzzle: Disentangling sources of bias at high frequency,
Journal of Financial Economics  {\bf 109} (2013), pp.~224--249.


\bibitem{alos07}
 E. Al\`os, J. A. Le\'on, and J. Vives,
On the short-time behavior of the implied volatility for jump-diffusion models with stochastic volatility,
Finance Stoch.  {\bf 11} (2007), pp.~571--589.

\bibitem{benth}
O. E. Barndorff-Nielsen, F. E. Benth, and A. E. D. Veraart, 
Modelling energy spot prices by volatility modulated L\'evy driven Volterra processes, 
Bernoulli  {\bf 19}  (2013), pp.~80--845.


\bibitem{gatheral2} 
C. Bayer, P. Friz, and J. Gatheral,  
Pricing under rough volatility,  
 Quantitative Finance {\bf 16} (2016), pp.~887--904.
 
\bibitem{beres} 
H. Berestycki, J. Busca, and I. Florent,  
Computing the implied volatility  in stochastic volatility models, 
Comm. Pure Appl. Math. {\bf 57} (2004), pp.~1352--1373. 

\bibitem{oksendal}
F. Biagini,  Y. Hu, B. \O ksendal, and T.  Zhang, 
Stochastic Calculus for Fractional Brownian Motion and Applications, 
Springer, London, 2008.

\bibitem{bollerslev}
T. Bollerslev, D. Osterrieder, N. Sizova, and G.  Tauchen,
Risk and return: Long-run relations, fractional cointegration, and return predictability,
Journal of Financial Economics {\bf 108} (2013), pp.~409--424.

\bibitem{carr}
P. Carr and L. Wu,
What type of processes underlies options? A simple robust test,
Journal of Finance {\bf 58} (2003), pp.~2581--2610.

\bibitem{cheridito03}
C. Cheridito,  H. Kawaguchi, and M. Maejima, 
Fractional Ornstein-Uhlenbeck processes, 
Electronic Journal of Probability {\bf  8} (2003), pp.~1--14.

\bibitem{viens1}  
A. Chronopoulou and F. G. Viens,
Estimation and pricing under long-memory stochastic volatility,
Annals of Finance {\bf  8}  (2012),  pp.~379--403.

\bibitem{viens2}  
A. Chronopoulou and F. G. Viens,
Stochastic volatility models with long-memory in discrete and continuous time,
Quantitative Finance {\bf 12} (2012), pp.~635--649. 

\bibitem{coutin} 
F. Comte, L. Coutin, and E. Renault,
Affine fractional stochastic volatility models,
Annals of Finance {\bf 8} (2010), pp.~337-378.

\bibitem{comte} 
F. Comte and E. Renault,
Long memory in continuous-time stochastic volatility models,
Mathematical  Finance  {\bf 8} (1998), pp.~291--323.  

\bibitem{cont} 
R. Cont,
Long range dependence in financial markets,
in Fractals in Engineering, edited by J. L\'evy V\'ehel and Evelyne Lutton, Springer, London, 2005, pp.~159--179.

\bibitem{corley}
S. Corley, J. Lebovits and J. L\'evy V\'ehel, 
Multifractional Stochastic volatility models,
 Mathematical Finance {\bf 24} (2014), pp.~364-402.
 
\bibitem{coutinb}
L. Coutin,  
An introduction to (stochastic) calculus with respect to fractional Brownian motion,  
in S\'eminaire de Probabilit\'es XL, Lecture Notes in Mathematics, Vol. 1899, Springer, 
2007, pp.~3--65.

\bibitem{taqqu} 
P. Doukhan, G. Oppenheim, and M. S. Taqqu,
Theory and Applications of Long-Range Dependence,
Birkh\"auser, Boston, 2003. 

\bibitem{duffie} 
D. Duffie, R. Pan, and K. Singleton, 
Transformation analysis and asset pricing for affine jump-diffusion, 
Econometrica {\bf  68} (2000), pp.~1343--1376.

\bibitem{engle} 
R. F. Engle, and A. J. Patton,
What good is a volatility model?,
Quantitative Finance {\bf 1} (2001), pp.~237--245.

\bibitem{figuer} 
J. E. Figueroa-L\'opez and S. Olafsson, 
Short-time expansions for close-to-the-money options under a L\'evy jump model with stochastic volatility, 
Finance and Stochastics {\bf 20} (2016), pp.~219--265.

\bibitem{fink13}
H. Fink, C. Kl\"uppelberg, and M. Z\"ahle, 
Conditional distributions of processes related to fractional Brownian motion, 
J. Appl. Prob. {\bf 50} (2013), pp.~166--183.

\bibitem{forde} 
M. Forde and H. Zhang, 
Asymptotics for rough stochastic volatility  and L\'evy models, 
preprint available at http://www.mth.kcl.ac.uk/\~{}fordem/. 

\bibitem{fouque00}
J. P. Fouque, G. Papanicolaou, and K. R. Sircar,
Derivatives in Financial Markets with Stochastic Volatility,
Cambridge University Press, Cambridge, 2000.

\bibitem{fouque03}
J. P. Fouque, G. Papanicolaou, K. R. Sircar, and K. Solna,
Singular perturbations in option pricing,
SIAM J. Appl. Math., {\bf 63} (2003), pp.~1648--1665.

\bibitem{smile}
J. P. Fouque, G. Papanicolaou,  K. R. Sircar,  and K. Solna,
Timing the smile,  
The Wilmott Magazine, March (2004).

\bibitem{fouque11}
J. P. Fouque, G. Papanicolaou,  K. R. Sircar,  and K. Solna,
Multiscale Stochastic Volatility for Equity, Interest Rate, and Credit Derivatives,
Cambridge University Press, Cambridge, 2011.
 
\bibitem{fukasawa11}
M. Fukasawa, 
Asymptotic analysis for stochastic volatility: martingale expansion, 
Finance and Stochastics {\bf 15} (2011), pp.~635--654.

\bibitem{fukasawa15}
M. Fukasawa,  
Short-time at-the-money skew and rough fractional  volatility,
arXiv:1501.06980.

\bibitem{gatheral1} 
J. Gatheral, T. Jaisson, and M. Rosenbaum,
Volatility is rough,  
arXiv:1410.3394.


\bibitem{jacquier}  
H. Guennoun, A. Jacquier, and P. Roome,
Aymptotic behaviour of the fractional Heston model, 
arXiv:1411.7653.
 
\bibitem{viens3}
A. Gulisashvili, F. Viens, and X. Zhang,
Small-time asymptotics for Gaussian self-similar stochastic volatility models,
arXiv:1505.05256.

\bibitem{henry} 
P. Henry-Labord\'ere,
Analysis, Geometry, and Modeling in Finance:
Advanced Methods in Option Pricing, Chapman \& Hall, Boca Raton, 2009.  
 
\bibitem{heston}
S. L. Heston, 
A closed-form solution for options with stochastic volatility with applicantions to bond  and currency options, 
The Review of Financial Studies {\bf 6} (1993), pp.~327--343. 

\bibitem{kaarakka}
T. Kaarakka and P. Salminen, 
On fractional Ornstein-Uhlenbeck processes,
Communications on Stochastic Analysis  {\bf 5} (2011), pp.~121--133. 

\bibitem{rodger}
R. Lee, 
 The Moment Formula for Implied Volatility at Extreme Strikes, 
 Mathematical Finance  {\bf 14} (2004),  pp.~469--480.

\bibitem{mandelbrot}
B. B. Mandelbrot and J. W. Van Ness,
Fractional Brownian motions, fractional noises and applications,
SIAM Review {\bf 10} (1968), pp.~422--437.

\bibitem{mendes2}
R. V. Mendes, M. J. Oliveira, and A. M. Rodrigues,
The fractional volatility model: An agent-based interpretation,
Physica A {\bf 387} (2008), pp.~3987--3994.

\bibitem{mendes1}
R. V. Mendes, M. J. Oliveira, and A. M. Rodrigues,
No-arbitrage, leverage and completeness in a fractional volatility model,
Physica A {\bf 419} (2015), pp.~470--478.

\bibitem{tankov} 
A. Mijatovic and P. Tankov,
A new look at short-term implied volatility in asset price models with jumps,
Mathematical Finance {\bf 26}  (2016), pp.~149--183.

\bibitem{muralev} 
A. A. Muralev,
Representation of franctional Brownian motion in terms of an infinite-dimensional Ornstein-Uhlenbeck process, 
Russ. Math. Surv. {\bf 66} (2011), pp.~439--441.


\bibitem{touzi} 
E. Renault and N. Touzi,
Option hedging  and implicit volatilities in a stochastic volatility model, 
Mathematical Finance {\bf 6} (1996), pp.~279--302.

\bibitem{mykland}  
C. D. Wang and P. A. Mykland,
The estimation of leverage effect with high-frequency data,
Journal of the American Statistical Association {\bf 109} (2014), pp.~197--215.

\end{thebibliography}
\end{document}